\ifdefined\submissionmode
  \documentclass[acmsmall,screen,review,anonymous]{acmart}
\else
  \documentclass[acmsmall,screen,nonacm]{acmart}
\fi

\usepackage{tabularx}
\usepackage{bbm}
\usepackage{bbold}
\usepackage{booktabs}
\usepackage{enumitem}
\usepackage{mathpartir}
\usepackage{subcaption}
\usepackage{tikz-cd}
\usepackage{xspace}
\ifdefined\finalmode
  \usepackage[final,commandnameprefix=ifneeded,commentmarkup=footnote]{changes}
\else
  \usepackage[commandnameprefix=ifneeded,commentmarkup=footnote]{changes}
\fi
\usepackage{changebar}
\ifdefined\finalmode\nochangebars\fi

\newcommand*{\GPS}{Galois slicing\xspace} % in case we change our mind about how to refer to this

\newcommand*{\enummargin}{8mm}

\newcommand{\AGDA}{\textit{(Formalised in Agda)}}

\makeatletter
\@ifclasswith{acmart}{anonymous}{
  
}{
  
}
\makeatother

\newcommand*{\etal}{{\it et al.}\xspace}

\definechangesauthor[name=Bob, color=purple]{bob}
\definechangesauthor[name=Roly, color=red]{roly}
\colorlet{deletedgray}{gray!70}
\setdeletedmarkup{\textcolor{deletedgray}{#1}}
\newcommand{\Deleted}[1]{\leavevmode\protect\cbstart\deleted{#1}\protect\cbend\xspace}
\newenvironment{DeletedBlock}{\par\cbstart\color{deletedgray}}{\cbend\par}
\newenvironment{AddedBlock}{\par\cbstart\color{blue}}{\cbend\par}
\ifdefined\finalmode
  \renewcommand{\Deleted}[1]{}%
  \renewcommand{\deleted}[2][]{}%
\fi

\newenvironment{salign*}
   {\par\nobreak\small\noindent\csname align*\endcsname}
   {\csname endalign*\endcsname}

\newcommand*{\defref}[1]{Definition~\ref{def:#1}}

\newcommand*{\exref}[1]{Example~\ref{ex:#1}}
\newcommand*{\figref}[1]{Figure~\ref{fig:#1}}

\newcommand*{\propref}[1]{Proposition~\ref{prop:#1}}
\newcommand*{\lemref}[1]{Lemma~\ref{lem:#1}}
\newcommand*{\remref}[1]{Remark~\ref{rem:#1}}
\newcommand*{\thmref}[1]{Theorem~\ref{thm:#1}}
\newcommand*{\secref}[1]{\S\ref{sec:#1}}
\newcommand*{\eqnref}[1]{Equation~\ref{eqn:#1}}
\newcommand*{\posscite}[1]{\citeauthor{#1}'s~\citeyearpar{#1}}

\newcommand*{\syn}[1]{\mathbf{#1}} % keywords
\newcommand*{\PrimTy}{\mathsf{PrimTy}}
\newcommand*{\PrimOp}{\mathsf{Op}}
\newcommand*{\List}{\mathrm{List}}
\newcommand*{\fold}{\mathsf{fold}}
\newcommand*{\nil}{\mathsf{nil}}
\newcommand*{\cons}{\mathsf{cons}}

\newcommand*{\tmApp}[2]{#1\;#2}

\newcommand*{\tmCase}[5]{\syn{case}\;#1\;\{#2.#3; #4.#5\}}
\newcommand*{\tmCons}[2]{\syn{cons}\;#1\;#2}

\newcommand*{\tmFoldList}[3]{\syn{fold}\;#1\;#2\;#3}
\newcommand*{\tmFst}[1]{\syn{fst}\;#1}
\newcommand*{\tmFun}[2]{\lambda#1.#2}
\newcommand*{\tmInl}[1]{\syn{inl}\;#1}
\newcommand*{\tmInr}[1]{\syn{inr}\;#1}
\newcommand*{\tmNil}{\syn{nil}}
\newcommand*{\tmPair}[2]{(#1,#2)}

\newcommand*{\tmSnd}[1]{\syn{snd}\;#1}
\newcommand*{\tmUnit}{()}

\newcommand*{\tyFun}{\to}

\newcommand*{\tyList}{\syn{list}}
\newcommand*{\tyProd}{\times}
\newcommand*{\tySum}{+}
\newcommand*{\tyUnit}{\syn{1}}

\newcommand*{\Op}{\mathsf{Op}}

\newcommand*{\emptyCxt}{\cdot}

\newcommand*{\adj}{\dashv}
\newcommand*{\transpose}[1]{#1^{\top}}
\newcommand*{\biinj}{\mathsf{i}}
\newcommand*{\biprod}{\oplus}
\newcommand*{\biproj}{\mathsf{p}}
\newcommand*{\cat}[1]{\mathcal{#1}}

\newcommand*{\comp}{\circ}

\newcommand*{\coprodM}[2]{[#1, #2]}
\newcommand*{\cotangents}{\tangents^*}

\newcommand*{\eval}{\varepsilon}
\newcommand*{\id}{\mathsf{id}}

\newcommand*{\internalHom}[2]{#1 \Rightarrow #2}

\newcommand*{\join}{\vee}

\newcommand*{\linearto}{\multimap}
\newcommand*{\meet}{\wedge}

\newcommand*{\disj}{\mathrel{\#}}

\newcommand*{\namedcat}[1]{\mathbf{#1}}

\newcommand*{\prodM}[2]{\langle #1, #2\rangle}

\newcommand*{\pullf}[1]{#1^{*}}
\newcommand*{\pushf}[1]{{#1_{*}}}

\newcommand*{\sem}[1]{\llbracket #1 \rrbracket}
\newcommand*{\cbn}[1]{\langle\kern-2pt\langle #1 \rangle\kern-2pt\rangle}

\newcommand*{\tangents}{\mathrm{T}}

\newcommand*{\zero}{0}

\newcommand*{\Two}{\mathbbm{2}}

\newcommand*{\RR}{\mathbbm{R}}

\newcommand*{\CMon}{\namedcat{CMon}}

\newcommand*{\Fam}{\namedcat{Fam}}
\newcommand*{\FinVect}{\namedcat{FDVect}}
\newcommand*{\FinRel}{\namedcat{FinRel}}

\newcommand*{\Mat}{\namedcat{Mat}}
\newcommand*{\SemiMod}{\namedcat{SemiMod}}
\newcommand*{\SDSemiMod}{\namedcat{SDSemiMod}}
\newcommand*{\conj}[1]{#1^{\dagger}}

\newcommand*{\Set}{\namedcat{Set}}

\newcommand*{\Man}{\namedcat{Man}}
\newcommand*{\GLR}{\namedcat{GLR}} % Grothendieck Logical Relations

\theoremstyle{acmdefinition}
\newtheorem{remark}{Remark}

\begin{document}

\title{Data Provenance as Automatic Differentiation}

\author{Robert Atkey}
\email{robert.atkey@strath.ac.uk}
\orcid{0000-0002-4414-5047}
\affiliation{%
  \institution{University of Strathclyde}
  \city{Glasgow}
  \country{UK}}

\author{Roly Perera}
\email{roly.perera@cl.cam.ac.uk}
\orcid{0000-0001-9249-9862}
\affiliation{%
  \institution{University of Cambridge}
  \city{Cambridge}
  \country{UK}
}
\additionalaffiliation{%
   \institution{University of Bristol}
   \city{Bristol}
   \country{UK}
}

\begin{abstract}
  Automatic differentiation (AD) computes the derivative of a program alongside the program itself,
  as a linear map between tangent spaces, propagated forwards or backwards along an execution. We present a
  semantic framework that models \emph{data provenance} via the same construction: taking scalars from a
  commutative semiring of dependency information rather than the real numbers, the derivative of a program
  becomes a linear map between spaces of approximations of its input and output. The choice of semiring
  determines the notion of provenance. Over the two-element Boolean algebra, the Jacobian of a program
  records which input positions each output position may depend on, and composing Jacobians forwards or
  backwards is \emph{dependency analysis} in the manner of forward- and reverse-mode AD. More generally, over
  distributive lattices the Jacobian and its transpose propagate dependency information forwards and
  backwards as a \emph{conjugate pair} of maps; when the lattice is a Boolean algebra, the two directions are
  moreover related by adjunction, recovering an approach called \emph{\GPS}. We interpret a higher-order total
  functional language in this framework, prove that every program of first-order type denotes such a Jacobian,
  and instantiate the semiring to obtain dependency tracking (Booleans), automatic differentiation (reals),
  and quantitative interval provenance (the tropical semiring) as examples. All results are formalised in
  Agda.
\end{abstract}
\maketitle

\section{Introduction}
\label{sec:introduction}

To audit any computational process, we need robust and well-founded notions of \emph{provenance} to track how
data are used. This allows us to answer questions like ``Where did these data come from?'', ``Why are these
data in the output?'' and ``How were these data computed?''. Provenance tracking has a wide range of
applications, from debugging and program comprehension~\cite{buneman95,cheney07} to improving reproducibility
and transparency in scientific workflows~\cite{kontogiannis08}. \emph{Program slicing}, first proposed
by~\citet{weiser81}, is a collection of techniques for provenance tracking that attempts to take a run of a
program and areas of interest in the output, and turn them into the subset of the input and the program that
were responsible for generating those specific outputs.

Underlying all of these questions is a more basic one: how does the output of a program respond when
its input changes? For programs that compute with real numbers, calculus provides a canonical answer: the
derivative, a linear map between tangent spaces describing how the output varies as the input varies, with
\emph{automatic differentiation} (AD) being a computational technique for computing such derivatives alongside
the program itself~\cite{siskind08,elliott18,vakar22}. For data provenance and dependency analysis, where the
inputs are database tuples or elements of a data structure, the question is often framed in a more qualitative
way: if we perturb the input at a given position, can the output change at all? Posing this question one input
at a time yields a vector of Boolean partial ``derivatives'', and for a program with several outputs, a
\emph{Boolean Jacobian}: a matrix whose entry at $(j, i)$ records whether output position $j$ may depend on
input position $i$. Consider multiplication: at the point $(x, y)$, the usual partial derivative $\partial(x
\cdot y)/\partial x$ is $y$, and its Boolean counterpart records that the output can vary with $x$ only when
$y \neq 0$. Like their numerical counterparts, Boolean Jacobians compose by matrix multiplication, with
conjunction and disjunction playing the role of multiplication and addition. This paper presents an approach
to dynamic dependency analysis where such Jacobians are evaluated alongside the program, forwards or
backwards, in the manner of forward- and reverse-mode AD, with the two directions related by matrix
transposition.

Existing approaches to program slicing are often tied to particular programming languages or implementations.
In this paper we develop this analogy into a general categorical approach to data provenance, in
which the derivative of a program is a linear map between spaces of approximations of its input and output,
with scalars drawn from a commutative semiring of dependency information. The choice of semiring determines
the analysis: the reals recover AD itself, while over distributive lattices the derivative and its transpose
form a \emph{conjugate pair}, propagating dependency information forwards and backwards; when the lattices are
Boolean algebras, the same maps can instead be presented as an adjunction, recovering the \GPS of Perera and
collaborators. Our categorical approach should provide a suitable setting for enabling ``automatic''
data provenance for a variety of programming languages, and is easily configured to use alternative
approximation strategies, including quantitative forms of slicing.

\subsection{Dependencies as Derivatives}
\label{sec:introduction:galois-slicing}

Consider a single run of a program on a particular input. Each part of the output is computed from
certain parts of the input, and recording these facts gives the \emph{dependency relation} of the run,
relating each output position to the input positions it depends on. This section works through a small example
showing how dependency relations behave as derivatives: they have a tabular representation as Boolean
matrices, compose using matrix multiplication, and can be evaluated forwards or backwards.

\begin{example}
  \label{ex:introduction-example}
  The following program is written in Haskell syntax \cite{haskell}. It filters a list of pairs of labels and
  quantities to those with a given label, and then sums the quantities weighted by a per-label price:
  \begin{displaymath}
    \begin{array}{l}
      \mathrm{total} :: \mathrm{Label} \to [(\mathrm{Label}, \mathrm{Rational})] \to (\mathrm{Label} \to \mathrm{Rational}) \to \mathrm{Rational} \\
      \mathrm{total}\,l\,\mathit{db}\,\mathit{price} = \mathrm{sum}\,[ \mathit{price}\,l \cdot q \mid (l',q) \leftarrow \mathit{db}, l \equiv l' ]
    \end{array}
  \end{displaymath}
  With $\mathit{db} = [(\mathsf{a}, 3), (\mathsf{b}, 1), (\mathsf{a}, -3)]$, where the third entry is a refund
  of the first, and $\mathit{price}\,\mathsf{a} = 2$ and $\mathit{price}\,\mathsf{b} = 5$, we have
  $\mathrm{total}\,\mathsf{a}\,\mathit{db}\,\mathit{price} = 2 \cdot 3 + 2 \cdot (-3) = 0$ and
  $\mathrm{total}\,\mathsf{b}\,\mathit{db}\,\mathit{price} = 5$.

  Suppose we are interested in how the output depends on the numerical parts of the input, for the
  query parameters $l = \mathsf{a}$ and $l = \mathsf{b}$. The three quantities in $\mathit{db}$
  and the two prices occupy five \emph{positions}, and the output a single position. In the run with $l =
  \mathsf{a}$, the output is computed from the quantities at the first and third positions and the price of
  $\mathsf{a}$; in the run with $l = \mathsf{b}$, from the second quantity and the price of
  $\mathsf{b}$. We depict each input position
  pointing to the output positions that depend on it:
  \begin{center}
    \begin{tikzpicture}[every node/.style={font=\small}]
      % total a
      \node (la) at (0,2.2) {$\mathrm{total}\;\mathsf{a}$};
      \node[anchor=east] (a1) at (1.0,1.6) {$(\mathsf{a},3)$};
      \node[anchor=east] (a2) at (1.0,0.8) {$(\mathsf{b},1)$};
      \node[anchor=east] (a3) at (1.0,0) {$(\mathsf{a},-3)$};
      \node[anchor=east] (a4) at (1.0,-0.8) {$\mathit{price}\,\mathsf{a} = 2$};
      \node[anchor=east] (a5) at (1.0,-1.6) {$\mathit{price}\,\mathsf{b} = 5$};
      \node (ao) at (2.8,0.3) {$0$};
      \draw[->] (a1.east) -- (ao);
      \draw[->] (a3.east) -- (ao);
      \draw[->] (a4.east) -- (ao);
      % total b
      \node (lb) at (6,2.2) {$\mathrm{total}\;\mathsf{b}$};
      \node[anchor=east] (b1) at (7.0,1.6) {$(\mathsf{a},3)$};
      \node[anchor=east] (b2) at (7.0,0.8) {$(\mathsf{b},1)$};
      \node[anchor=east] (b3) at (7.0,0) {$(\mathsf{a},-3)$};
      \node[anchor=east] (b4) at (7.0,-0.8) {$\mathit{price}\,\mathsf{a} = 2$};
      \node[anchor=east] (b5) at (7.0,-1.6) {$\mathit{price}\,\mathsf{b} = 5$};
      \node (bo) at (8.8,-0.4) {$5$};
      \draw[->] (b2.east) -- (bo);
      \draw[->] (b5.east) -- (bo);
    \end{tikzpicture}
  \end{center}
    Tabulating each relation, with a row for each output position and a column for each input position
  (ordered as the three quantities and then the two prices), gives the \emph{Boolean Jacobian} of the run,
  shown here beside the numerical Jacobian of differential calculus, taken at the same point:
  \begin{displaymath}
    \partial_{\Two}(\mathrm{total}\,\mathsf{a}\,\mathit{db}\,\mathit{price}) =
    \begin{pmatrix} 1 & 0 & 1 & 1 & 0 \end{pmatrix}
    \qquad
    \partial(\mathrm{total}\,\mathsf{a}\,\mathit{db}\,\mathit{price}) =
    \begin{pmatrix} 2 & 0 & 2 & 0 & 0 \end{pmatrix}
  \end{displaymath}
  Like their numerical counterparts, these matrices compose. Extensionally, we can think of
  $\mathrm{total}\,l$ as a sum applied to per-row products applied to a selection, and the Jacobian of the run
  as the product of the Jacobians of those stages: a chain rule for dependencies. For the run with $l =
  \mathsf{a}$, the Boolean pipeline is
  \begin{displaymath}
    \underbrace{\vphantom{\begin{pmatrix} 0 \\ 0 \\ 0 \end{pmatrix}}\begin{pmatrix} 1 & 1 \end{pmatrix}}_{\text{sum}}
    \;
    \underbrace{\vphantom{\begin{pmatrix} 0 \\ 0 \\ 0 \end{pmatrix}}\begin{pmatrix} 1 & 0 & 1 \\ 0 & 1 & 1 \end{pmatrix}}_{\text{multiply}}
    \;
    \underbrace{\begin{pmatrix} 1 & 0 & 0 & 0 & 0 \\ 0 & 0 & 1 & 0 & 0 \\ 0 & 0 & 0 & 1 & 0 \end{pmatrix}}_{\text{select}}
    \;=\;
    \begin{pmatrix} 1 & 0 & 1 & 1 & 0 \end{pmatrix}
  \end{displaymath}
  and the numerical pipeline is
  \begin{displaymath}
    \underbrace{\vphantom{\begin{pmatrix} 0 \\ 0 \\ 0 \end{pmatrix}}\begin{pmatrix} 1 & 1 \end{pmatrix}}_{\text{sum}}
    \;
    \underbrace{\vphantom{\begin{pmatrix} 0 \\ 0 \\ 0 \end{pmatrix}}\begin{pmatrix} 2 & 0 & 3 \\ 0 & 2 & -3 \end{pmatrix}}_{\text{multiply}}
    \;
    \underbrace{\begin{pmatrix} 1 & 0 & 0 & 0 & 0 \\ 0 & 0 & 1 & 0 & 0 \\ 0 & 0 & 0 & 1 & 0 \end{pmatrix}}_{\text{select}}
    \;=\;
    \begin{pmatrix} 2 & 0 & 2 & 0 & 0 \end{pmatrix}
  \end{displaymath}
  Reading right to left: selection maps the five input positions to the three inputs used on this run (the two
  $\mathsf{a}$-quantities and $\mathit{price}\,\mathsf{a}$); the multiply stage computes the two per-row
  products, each depending on its quantity (with price as coefficient) and on the price (with quantity as
  coefficient); and the sum merges the two products. The two pipelines differ in the price column of the
  multiply stage, where the numerical coefficients $3$ and $-3$ cancel through the sum while their Boolean
  counterparts join as $1 \vee 1 = 1$.

  In the composite Jacobians, the quantity entries of $\partial$ are the relevant price, since the derivative
  of a product is the other factor, and the corresponding $\partial_{\Two}$-entries record which of them are
  nonzero, so on multiplication the two Jacobians agree. At the entry for $\mathit{price}\,\mathsf{a}$ they
  differ: $\partial\,\mathrm{total}\,\mathsf{a} / \partial\,\mathit{price}\,\mathsf{a} = 3 + (-3) = 0$, since
  the refund cancels the purchase and the total genuinely does not depend on that price; the Boolean Jacobian
  records $1$, because dependency information combines by disjunction and there is no negative dependency to
  cancel with. Boolean Jacobians can thus report more nonzero entries than their numerical counterparts, but
  never fewer; this over-approximation is inherent to tracking binary dependence.

  As in differential calculus, the Jacobian is taken at a point; the run with $l = \mathsf{b}$, a different
  point, has its own pipeline of stage Jacobians, with composites:
  \begin{displaymath}
    \partial_{\Two}(\mathrm{total}\,\mathsf{b}\,\mathit{db}\,\mathit{price}) =
    \begin{pmatrix} 0 & 1 & 0 & 0 & 1 \end{pmatrix}
    \qquad
    \partial(\mathrm{total}\,\mathsf{b}\,\mathit{db}\,\mathit{price}) =
    \begin{pmatrix} 0 & 5 & 0 & 0 & 1 \end{pmatrix}
  \end{displaymath}
  At this point nothing cancels, and the two Jacobians agree on which entries are nonzero.

  In setting up the positions, we chose to let only the numbers (quantities and prices) vary, keeping the
  labels and the structure of the list itself fixed. Other choices are also useful, and one of the aims of
  this work is to clarify how to choose a dependency structure appropriate for different tasks by selecting an
  appropriate semiring of dependency information. We elaborate on this further in \secref{approx-as-tangents}.

  As a matrix, the Boolean Jacobian acts on Boolean vectors indexed by input positions, and this action
  evaluates dependencies \emph{forwards}: a vector selects input positions, and its image under the Jacobian
  selects the output positions that depend on them. Such a vector can also be displayed in the shape of the
  data the program consumes, as an \emph{approximation} of the input, keeping the numbers at selected
  positions and replacing the rest by $\bot$, with the price function shown as the pair $\langle
  \mathit{price}\,\mathsf{a}, \mathit{price}\,\mathsf{b} \rangle$. Writing
  $\partial_{\Two}(\mathrm{total}\,l\,\mathit{db}\,\mathit{price})_f$ for this forward action, we have
  $\partial_{\Two}(\mathrm{total}\,\mathsf{b}\,\mathit{db}\,\mathit{price})_f([(\mathsf{a},\bot),(\mathsf{b},1),(\mathsf{a},\bot)],
  \langle \bot, \bot \rangle) = 5$, since the output depends on the entry $(\mathsf{b},1)$, whereas
  $\partial_{\Two}(\mathrm{total}\,\mathsf{a}\,\mathit{db}\,\mathit{price})_f$ applied to the same selection
  returns $\bot$, because that run consulted no $\mathsf{b}$-labelled entries.

  The transpose of the Jacobian evaluates dependencies \emph{backwards}: given a selection of
  output positions, it returns the input positions on which they depend. Writing
  $\partial_{\Two}(\mathrm{total}\,l\,\mathit{db}\,\mathit{price})_r$ for this reverse action, the fully
  approximated output $\bot$ depends on none of the input:
  \begin{displaymath}
    \partial_{\Two} (\mathrm{total}\,l\,\mathit{db}\,\mathit{price})_r(\bot) =
    ([(\mathsf{a},\bot), (\mathsf{b}, \bot), (\mathsf{a}, \bot)], \langle \bot, \bot \rangle)
  \end{displaymath}
  for both $l = \mathsf{a}$ and $l = \mathsf{b}$. If instead we keep the output unapproximated, then the two
  runs' backwards actions return different results:
  \begin{displaymath}
    \begin{array}{r@{\;=\;([\,}c@{,\;}c@{,\;}c@{\,],\;\langle\,}c@{,\;}c@{\,\rangle)}}
      \partial_{\Two} (\mathrm{total}\,\mathsf{a}\,\mathit{db}\,\mathit{price})_r(0)
        & (\mathsf{a},3) & (\mathsf{b},\bot) & (\mathsf{a},-3) & 2 & \bot \\
      \partial_{\Two} (\mathrm{total}\,\mathsf{b}\,\mathit{db}\,\mathit{price})_r(5)
        & (\mathsf{a},\bot) & (\mathsf{b},1) & (\mathsf{a},\bot) & \bot & 5
    \end{array}
  \end{displaymath}
  Pieces of the input that were not used are replaced by $\bot$. As we expect, the run with label $\mathsf{a}$
  depends on the entries in the database labelled with $\mathsf{a}$, and likewise for the run with label
  $\mathsf{b}$. Note that the backward action for $\mathsf{a}$ retains $\mathit{price}\,\mathsf{a}$, even
  though the numerical derivative there is $0$: read backwards, too, the dependency Jacobian cannot see that
  contributions cancelled. Asking \emph{how much} an output can change as an input varies, rather than whether
  it can move at all, calls for a more quantitative choice of scalars; we develop this in
  \secref{approx-as-tangents:richer-semirings}.

\end{example}

In a simple query like this, it is easy to work out the dependency relationship between the input and output.
However, the benefit of language-based approaches is that they are {\em automatic} for all programs, no
matter how complex the relationship between input and output. Moreover, by changing what we mean by
``approximation'' we can compute a range of different information about a program.

\subsection{Dependency Analysis and Automatic Differentiation}

The backward maps of the previous section compute a form of \emph{data provenance}: for each part of
the output, the part of the input needed to explain it. Dependency information of this kind is the common
currency of a range of existing techniques, from program slicing~\cite{weiser81} to provenance for database
queries~\cite{buneman95,cheney07}, including semiring provenance~\cite{green07}. Closest
to our setting is the \GPS of Perera and
collaborators~\cite{perera12a,perera16d,ricciotti17}, which organises the forward and backward maps of a run
into a Galois connection between lattices of approximations.

Many forms of dynamic provenance analysis, \GPS included, work by recording a trace of each
execution and computing the backwards analysis by re-running over the trace. A denotational account would
instead bake the analysis into the semantics of the program itself, rather than provide it as a separately
defined ``backwards evaluation'' operation.
The differential reading of dependency analysis, developed in \secref{approx-as-tangents}, points to
a way to obtain such an account, by analogy with \emph{automatic
differentiation}~\cite{siskind08,elliott18,vakar22}; let us make the correspondence explicit.

\begin{itemize}[leftmargin=\enummargin]
\item For dependency analysis, every value has an associated join-semilattice of \emph{approximations}, with
least element $\bot$ (the empty selection). For differentiable programs, every point has an associated vector
space of {\em tangents}.
\item For dependency analysis, every program has an associated forward approximation map that takes
approximations of the input to the approximations of the output that depend on them. This map {\em preserves
joins and $\bot$}. For differentiable programs, every program has a forward derivative that takes tangents of
the input to tangents of the output. The forward derivative map is {\em linear}, so it preserves addition of
tangents and the zero tangent.
\item For dependency analysis, every program has an associated backward approximation map that takes
approximations of the output back to least approximations of the input that they depend on. This map also {\em
preserves joins and $\bot$}. For differentiable programs, every program has a reverse derivative that takes
\emph{cotangents} of the output (linear maps from tangents to scalars) to cotangents of the input. This map is
again {\em linear}.
\item In both settings, the forward and backward maps are related by being each other's transpose. The
transpose makes sense because the spaces are \emph{self-dual}: tangent vectors pair with tangent vectors via
the dot product, identifying cotangents with tangents, and selections of positions pair by the same dot
product formula, with joins and meets in place of sums and products.
\end{itemize}

Given this close connection between dependency analysis and differentiable programming, we can take structures
intended for modelling automatic differentiation, such as V{\'a}k{\'a}r's CHAD framework and use them to model
dependency analysis. This will enable us to generalise and expand the scope of automatic differentiation to
act as a foundation for data provenance in a wider range of computational settings.

\subsection{Outline and Contributions}

Our main contribution is to show that data provenance can be understood as differentiation over a commutative
semiring of dependency information. In \secref{approx-as-tangents} we develop the first-order picture, for
programs over tuples of atomic data: derivatives as matrices over a semiring $S$, with the transpose relating
forward and backward dependency analysis, and with successively stronger conditions on $S$ yielding
join-preserving maps, conjugate pairs and, for Boolean algebras, the Galois connections of \GPS; ordinary AD
occupies the case $S = \mathbb{R}$.

The following table summarises the correspondence developed over the course of the paper:
\vspace{1.5mm}
\begin{center}
  \small
  \begin{tabular}{lll}
    \toprule
    & \textbf{Differentiable programming} & \textbf{Dependency analysis} \\
    \midrule
    scalars & real numbers & commutative semiring $S$ \\
    space at a point & vector space of tangents & semimodule of approximations over $S$ \\
    identification with the dual & finite dimensionality & chosen self-duality \\
    derivative at a point & Jacobian matrix & $S$-weighted relation as matrix over $S$ \\
    chain rule & matrix multiplication & matrix multiplication over $S$ \\
    forward mode & push tangents forward & propagate dependencies forward \\
    reverse mode & pull cotangents back & propagate dependencies backwards \\
    forward/backward relation & transpose & transpose via self-dualities \\
    \bottomrule
  \end{tabular}
\end{center}
\vspace{1.5mm}

In \secref{models-of-total-gps} we extend to sum types using the category of families construction and,
following V{\'a}k{\'a}r \etal{}'s CHAD framework~\cite{vakar22,nunes2023}, build models of a higher-order
language in which every program is interpreted together with its family of derivatives. We apply this to a
concrete higher-order language in \secref{language} and demonstrate the model on a number of examples in \secref{examples}, showing
how parameterising on $S$ controls the dependency structure associated with data, something that was
``hard-coded'' in previous presentations of \GPS. We prove a correctness property in
\secref{definability}: programs of first-order type have equal interpretations in standard semantics and a dependency tracking semantics, even when they use higher-order functions internally. \secref{related-work} and \secref{conclusion} discuss additional related and future work.

% \begin{enumerate}[leftmargin=\enummargin]
% \item We explain how the CHAD framework of Vákár et al. can be adapted to give a general categorical framework for \GPS.
% \item Using our adapted CHAD framework, we can explain how type structure can be used to control the approximation lattices associated with data points, something that was ``hard coded'' in previous presentations of \GPS.
% \item With the benefit of our abstract setting, we can relate \GPS to other parts of denotational semantics. In particular, we show that there is a close connection with Stable Domain Theory, a proposed solution to capturing sequentiality by recording more intensional information about programs' sensitivity to approximations.
% \end{enumerate}

We have formalised our major results in Agda, resulting in an executable implementation built directly from the categorical constructions that we have used to compute the examples in \secref{examples}. Please consult the file \texttt{everything.agda} in the supplementary material.

% \begin{enumerate}
% \item We explain how the CHAD framework of Vákár et al. can be adapted to give a general categorical framework for \GPS.
% \item Using our adapted CHAD framework, we can explain how type structure can be used to control the approximation lattices associated with data points, something that was ``hard coded'' in previous presentations of \GPS.
% \item With the benefit of our abstract setting, we can relate \GPS to other parts of denotational semantics. In particular, we show that there is a close connection with Stable Domain Theory, a proposed solution to capturing sequentiality by recording more intensional information about programs' sensitivity to approximations.

% We have \anonymise{\href{https://github.com/bobatkey/approx-diff}{formalised this work in Agda}}{formalised this work in Agda}. Not only does this mean that the constructions and proofs have been checked, but also the construction is executable and can be run on small examples. We are also planning to use Agda's JavaScript backend to produce a demo.

\section{Tangent Spaces as Semimodules}
\label{sec:approx-as-tangents}

We motivate our approach by showing how to combine ideas from differential geometry and linear algebra to
reconstruct the dependency analyses of \secref{introduction} in a denotational setting.

\subsection{Manifolds, Smooth Functions, and Automatic Differentiation}
\label{sec:approx-as-tangents:autodiff}

The general study of differentiable functions takes place on \emph{manifolds}, topological spaces that
``locally'' behave like an open subset of the Euclidean space $\RR^n$. The spaces $\RR^n$ themselves are
manifolds, but so are ``non-flat'' examples such as $n$-spheres and yet more exotic spaces. Every point $x$ in
a manifold $M$ has an associated \emph{tangent vector space} $\tangents_x(M)$ consisting of linear
approximations of curves on the manifold passing through $x$. Each point also has a \emph{cotangent vector
space} $\cotangents_x(M) = \tangents_x(M) \linearto \RR$. The tangent and cotangent spaces are finite
dimensional, so in the presence of a chosen basis they are canonically isomorphic. In the case when the
manifold is $\RR^n$, then every tangent space is isomorphic to $\RR^n$ as well.

Smooth functions $f$ between manifolds $M$ and $N$ are functions on their points that are locally
differentiable on $\RR^n$. Manifolds and smooth functions form a category $\Man$. Each smooth function induces
maps of the (co)tangent spaces:
\begin{itemize}[leftmargin=\enummargin]
\item The \emph{forward derivative} (tangent map, pushforward) $\pushf{f}_x$ is a linear map $\tangents_x(M)
\linearto \tangents_{f(x)}(N)$. In the Euclidean case when $M = \RR^m$ and $N = \RR^n$, the tangent map can be
represented by the Jacobian matrix of partial derivatives of $f$ at $x$.
\item The \emph{backward derivative} (cotangent map, pullback) $\pullf{f}_x$ is a linear map
$\cotangents_{f(x)}(N) \linearto \cotangents_x(M)$. In the Euclidean case, the backward derivative is
represented by the transpose of the Jacobian of $f$ at $x$.
\end{itemize}

\begin{remark}[Chain Rule]
  \label{rem:chain-rule}
  A useful property of derivative maps is that they compose according
  to the chain rule. Suppose that $f : M \to N$ and $g : N \to K$ are
  smooth functions. Then for any $x \in M$, we have:
  \begin{itemize}
  \item $\pushf{(g \circ f)}_x = \pushf{g}_{f(x)} \circ \pushf{f}_x : \tangents_x(M) \linearto \tangents_{g(f(x))}(K)$
  \item $\pullf{(g \circ f)}_x = \pullf{f}_x \circ \pullf{g}_{f(x)} : \cotangents_{g(f(x))}(K) \linearto \cotangents_x(M)$
  \end{itemize}
  The chain rule has the practical effect that we can compute
  derivative maps of $f$ and $g$ independently and compose them,
  instead of the potentially more difficult task of computing the
  derivative maps of $g \circ f$. As we shall see below, stable maps
  also obey a chain rule, and this forms the basis of the general
  categorical approach to differentiability that we describe in
  \secref{models-of-total-gps}.
\end{remark}

Computing the forward and backward derivatives of smooth functions $f$ has many applications of practical
interest. For example, computation of the reverse derivative is of central interest in machine learning by
gradient descent, the main technique used to train deep neural networks~\cite{rumelhart88,goodfellow16}.

Derivatives can be computed numerically by computing $f$ on small perturbations of its input, or symbolically
by examining a closed-form representation of $f$. However, a more common and practical technique is to use
\emph{automatic differentiation}, where a program computing $f$ is instrumented to produce (a representation
of) the forward and/or backward derivative as a side-effect of producing the output~\cite{linnainmaa76}. This
has led to the area of differentiable programming, where programming languages and their implementations are
specifically designed to admit efficient automatic differentation
algorithms~\cite{jax2018github,abadi16,elliott17,sigal24}.

\subsection{Derivatives over a Commutative Semiring}
\label{sec:approx-as-tangents:semiring}

In the Euclidean case, the differential picture above is plain matrix algebra: with a chosen basis, tangent
and cotangent spaces are both $\RR^n$, the forward derivative at a point is the Jacobian matrix, and the
backward derivative is its transpose. Once the picture is presented in coordinates like this, the real numbers
play no further role beyond supplying the scalars; since dependency information is typically indexed by
positions in a value, which then serve as a coordinate basis, our first move is to keep the matrix algebra and
change the scalars. Thus throughout this section, we fix a commutative semiring $S = (S, +, 0, \cdot, 1)$,
whose elements we think of as atoms of dependency information, with $\cdot$ combining information along a path
from an input to an output and $+$ combining information across parallel paths.

Our second move is one of abstraction: rather than axiomatise the matrices of \secref{introduction} directly,
we identify the minimal structure required to capture what made that story work: a space of dependency
information for each value, linear maps between spaces for the forward analysis, and an identification of each
space with its dual, the extra structure needed to support an analogue of matrix transposition, and hence a
backward analysis. The role of vector spaces, which are defined over fields, is played by \emph{semimodules}:
commutative monoids $(M, +, 0)$ equipped with an action $a \cdot x$ of the scalars, satisfying the usual laws.
Linear maps preserve $+$, $0$ and the action, and form a category $\SemiMod(S)$; pointwise addition of linear
maps makes each hom-set a commutative monoid, with composition bilinear. Every semimodule has a \emph{dual}
$M^* = M \linearto S$ of linear maps into the scalars (categorically, the internal hom into the scalars
object, which is the unit of the tensor product of semimodules), and every linear map $f : M \linearto N$ has
a \emph{transpose} $\transpose{f} = (- \comp f) : N^* \linearto M^*$, contravariantly.

The transpose is almost the backward analysis we are after: it runs in the opposite direction to $f$, but
between the duals $N^*$ and $M^*$, the cotangent spaces, rather than between the tangent spaces $M$ and $N$
themselves. What we want is a map $N \linearto M$, running dependency information back through the same spaces
the forward map runs it through, as the transposed matrix did in \secref{introduction}. For that, each space
must be identified with its dual. Whereas a finite-dimensional vector space always admits such an
identification (choose a basis), a semimodule need not; we therefore make the identification part of the data.
\begin{definition}[Self-dual semimodule]
  \label{def:self-dual-semimodule}
  A \emph{self-dual} semimodule is a pair $\langle M, i \rangle$ of a semimodule and a chosen isomorphism
  $i : M \cong M^*$. The \emph{pairing} induced by $i$ is $\langle x, y \rangle = i(x)(y)$.
\end{definition}
The pairing is an $S$-valued measure of the extent to which $x$ and $y$ overlap; its prototype is the
dot product of vectors (the pairing induced by the self-duality of $S^n$, as we will see below). Extending the
analogy, we call $x$ and $y$ \emph{orthogonal} when $\langle x, y \rangle = 0$.
\begin{definition}[Conjugate]
  \label{def:conjugate}
  Between self-dual semimodules $\langle M, i \rangle$ and $\langle N, j \rangle$, the \emph{conjugate} of
  a linear map $f : M \linearto N$ is $\conj{f} = i^{-1} \comp \transpose{f} \comp j : N \linearto M$, the
  unique linear map satisfying
  \begin{displaymath}
    \langle f(x), y \rangle = \langle x, \conj{f}(y) \rangle
  \end{displaymath}
\end{definition}
The conjugate is our backward analysis, and its defining equation makes precise the relationship with
the forward one alluded to in \secref{introduction}. Conjugation preserves identities, reverses
composition ($\conj{g \comp f} = \conj{f} \comp \conj{g}$) and is involutive, so backward analyses compose
according to the chain rule run in reverse, just as reverse derivatives do (\remref{chain-rule}). Self-dual
semimodules and linear maps form a dagger category $\SDSemiMod(S)$. A morphism is not
required to relate the chosen self-dualities of its domain and codomain; the choice determines the
conjugate, not which maps exist. $\SDSemiMod(S)$ is thus equivalent to the full subcategory of
$\SemiMod(S)$ on the semimodules that admit a self-duality.

The matrices of \secref{introduction} now reappear as concrete instances. The \emph{free} semimodule $S^n$ on
$n$ generators has as elements the $n$-length vectors of scalars, with standard basis vectors $e_i$ and every
$v \in S^n$ decomposing as $v = \sum_i v_i \cdot e_i$. Currying the \emph{dot product} $\langle u, v \rangle =
\sum_i u_i \cdot v_i$ gives an isomorphism $S^n \cong (S^n)^*$, since pairing with the basis vectors extracts
components, $\langle e_i, v \rangle = v_i$; so free semimodules are self-dual. By basis decomposition, a
linear map $S^m \linearto S^n$ is exactly an $n \times m$ matrix over $S$, with composition as matrix
multiplication and the conjugate as the familiar matrix transpose, $(\conj{M})_{ji} = M_{ij}$.

We write $\Mat(S)$ for the skeletal category whose objects are the natural numbers and whose morphisms $m \to
n$ are the $n \times m$ matrices. The Boolean Jacobians from \secref{introduction} live here: a dependency
relation between $m$ input positions and $n$ output positions is exactly a morphism $m \to n$ in $\Mat(\Two)$,
where $\Two$ is the Boolean semiring with $+ = \vee$ and $\cdot = \wedge$, and the chain rule for dependencies
is composition in $\Mat(\Two)$, or equivalently \emph{relational} composition under the equivalence of
categories $\Mat(\Two) \simeq \FinRel$. The dot product can be interpreted as ``overlap'', becauses $\langle
u, v \rangle = \top$ exactly when the selections $u$ and $v$ share a position, i.e.\ are not disjoint.

The trivial semimodule is a zero object in $\SDSemiMod(S)$, and self-dual semimodules are closed under the
\emph{biproduct} $M \biprod N$, an object which is both a product and a coproduct, with the injections
conjugate to the projections, $\biinj_i = \conj{\biproj_i}$. The coincidence of products and coproducts is
characteristic of such linear settings, and means $\SDSemiMod(S)$ on its own cannot interpret programs over
sum types; we return to this in \secref{models-of-total-gps}.

\subsection{Dependency Structures over Semirings}
\label{sec:approx-as-tangents:richer-semirings}

The framework so far only asks for commutativity of the scalars, which makes the dual $M^*$ an $S$-semimodule
again, so that the transpose remains a map between $S$-semimodules. Any commutative semiring $S$ thus yields a
category $\SDSemiMod(S)$ with forward derivatives and their conjugates, related through the pairing. When $S$
is the reals, the conjugate is the familiar adjoint of linear algebra. In a dependency-tracking setting,
however, the relationships of interest are typically granular: we ask whether, or to what degree, an output
depends on an input, and want to compare one answer with another. This calls for an order on dependency
information, which additional axioms on $S$ provide.

\subsubsection{Conjugate pairs and conjugacy analyses}

Suppose first that $1$ is an additive top in $S$ ($1 + a = 1$). Addition is then
idempotent, meaning $a + a = a$. Idempotent addition induces an order: $S$ becomes a join-semilattice,
with $a \leq b$ iff $a + b = b$, joins given by $+$ and least element $0$. This structure lifts pointwise to every semimodule,
and every linear map automatically preserves joins and $0$. Dependency
information is now ordered by informativeness: $0$ means ``no dependency'', and the forward and backward
analyses are monotone, taking more informative inputs to more informative outputs and vice-versa.

If multiplication is also idempotent ($a \cdot a = a$), then $S$ is a bounded distributive lattice,
with multiplication as the meet and distributivity given by the semiring. This opens the way to reading the
two analyses in terms of \emph{disjointness}: where meets exist, define $x \disj y$ when $x \meet y = 0$, and
if disjointness agrees with the orthogonality given by the pairing, the defining equation of the conjugate
(\defref{conjugate}) says that $f$ and $\conj{f}$ exchange disjointness: $f(x) \disj y$ iff $x \disj
\conj{f}(y)$. Let us therefore consider the case where each semimodule has meets, preserved by the scalar action and
agreeing with orthogonality in this sense; meets, unlike joins, do not lift from $S$, so this is additional
structure. Maps that exchange disjointness in this way have a standard name.
\begin{definition}[Conjugate pair]
  \label{def:conjugate-pair}
  Join-preserving maps $f : L \to M$ and $g : M \to L$ between bounded lattices form a \emph{conjugate
  pair} when $f(x) \disj y$ iff $x \disj g(y)$ for all $x \in L$ and $y \in M$.
\end{definition}
The notion originates in relation algebra~\cite{jonsson51}, where the conjugate of composition with a
relation is composition with its converse. The scalars $S$, viewed as a semimodule over itself, qualify
automatically for this reading, and biproducts (including the empty biproduct $S^0$) inherit the structure
componentwise, giving the linear maps between free semimodules $S^n$, along with their conjugates, as the canonical
instances of conjugate pairs.

The relevance to us here is that conjugate pairs also underpin the \emph{cognacy} (common ancestry)
analyses of work on linked visualisations~\cite{perera22,bond25}. These dependency analyses work by composing
a join-preserving dependency analysis with its conjugate, yielding a map sending a selection of outputs to its
\emph{related outputs}, those sharing a dependency on some part of the input; selections in one chart can then
be linked to ``cognate'' selections in another.

\subsubsection{Galois slicing}

Suppose in addition that every scalar has a complement: an operation $\neg$ with $a \meet \neg a =
0$ and $a \join \neg a = 1$, making $S$ a Boolean algebra. Negation supports an alternative presentation of
conjugates as adjoints. Composing a map with complements on either side preserves its direction but exchanges
join-preservation for meet-preservation; the \emph{De Morgan dual} $\neg \comp \conj{f} \comp \neg$ of
the conjugate is thus meet-preserving, and is precisely the upper adjoint of $f$ (upper adjoints
preserve meets, dually to lower adjoints and joins). Forward and backward analysis therefore form a
\emph{Galois connection} between the lattices of approximations.
\begin{definition}[Galois connection]
  \label{def:galois-connection}
  Suppose $X$ and $Y$ are posets. A \emph{Galois connection} $f \adj g: X \to Y$ is a pair of monotone
  functions $f: Y \to X$ and $g: X \to Y$ satisfying $y \leq g(x) \iff f(y) \leq x$ for any $x \in X$ and $y
  \in Y$. Since a Galois connection is also an adjunction, we refer to $f$ as the left adjoint and $g$ as the
  right adjoint.
\end{definition}
This is precisely the setting of \GPS~\cite{perera12a,perera16d,ricciotti17}, recovered here as the
Boolean instance of the general semimodule picture. Galois connections as such presuppose none of the
algebraic structure of this section, being definable between arbitrary posets; De Morgan duality is simply how
the adjoint presentation arises here, where the conjugate is the primary notion. One orientation to keep
straight: the conjugate as presented here, $\conj{f}$, runs backwards, whereas the meet-preserving map of \GPS
runs forwards. The two presentations agree because $\conj{-}$ is involutive: instantiating the construction at
$\conj{f}$ in place of $f$ gives exactly the connection of \GPS, with meet-preserving forward map $\neg \comp
f \comp \neg$.

The join-preserving and meet-preserving maps answer different questions. In the forwards direction, the
join-preserving map $f$ has a \emph{necessity}-flavoured interpretation: which outputs may depend on a given
part of the input. Its meet-preserving counterpart $\neg \comp f \comp \neg$ reads as \emph{sufficiency}:
which outputs a given input selection determines outright. Notably, in the Galois slicing work the
meet-preserving forward map is only used to establish the existence of the join-preserving backward map; it
seems to lack any direct dependency analysis applications of its own.

\subsubsection{Perturbation bounds}
\label{sec:approx-as-tangents:perturbation-bounds}

Dependency analysis over the Booleans records whether an output can change when an input does;
\exref{introduction-example} raised the more fine-grained question of how much it
can change. Numerical analysis studies this question as \emph{propagation of error}: given bounds on how far
each input may be perturbed, derive a bound on the output~\cite{high:ASNA2}. Propagation of error is again a
choice of semiring, this time the \emph{tropical semiring} $(\mathbbm{Q}_{\geq 0} \cup \{\infty\}, \min,
\infty, +, 0)$~\cite{simon88}. A scalar is a non-negative rational (or $\infty$) bounding the extent to which
a value may change from its value in the run, with $\infty$ meaning unconstrained. Semiring addition is
$\min$, so that combining bounds for parallel paths keeps the tighter bound, and multiplication is numerical
$+$, accumulating bounds along a path. Addition is idempotent, so the analyses are monotone as in
\secref{approx-as-tangents:richer-semirings}: smaller bounds are more informative, and the semiring zero
$\infty$ means ``no information''.

A Jacobian entry $c$ now acts as a translation: if the input at that position changes by at most
$\delta$, with the others held fixed, the output changes by at most $\delta + c$. Addition's entries are $0$: a
bound on either argument passes through undisturbed. Multiplication, however, scales perturbations: if $x$
changes by $\delta$ then $x \cdot y$ changes by $|y|\,\delta$, and no translation bounds a scaling; its
entries are $\infty$. This is sound, but reports nothing: propagation of absolute error is blind to
multiplication.

Measuring perturbations multiplicatively instead keeps $\min$ as the join of bounds but accumulates
along paths by numerical product: $(\mathbbm{Q}_{\geq 0}
\cup \{\infty\}, \min, \infty, \times, 1)$, with scalars now read as factors. This is the
multiplicative counterpart of the previous presentation, related to it by logarithms over the reals, and the
two are the classical dichotomy between absolute and relative error~\cite[\S\S 1.6--1.7]{high:ASNA2}.
Multiplication now has unit entries, since relative perturbations pass through a product unchanged, while an
addition $x + y$ acquires the entry $|x| / |x + y|$: the factor by which it amplifies relative error, its
(relative) \emph{condition number}.

In the refund example the final sum adds the products $6$ and $-6$, so the amplification is $|6| / |6 +
(-6)|$: an infinite condition number, known in numerical analysis as catastrophic cancellation. For
non-negative data the factors are at most $1$, so sums never amplify; the refund's negative entry makes the
infinite factor possible. Where the Boolean analysis could only report a dependency that was not really
there, the relative analysis says how badly the bound degrades. The difference between the two is a choice of
semiring, and making that choice is part of specifying what the analysis is to observe.

\subsubsection{Qualitative derivatives}
\label{sec:approx-as-tangents:qualitative}

The Boolean analysis reports a dependency at a cancellation because it forgets signs. The \emph{sign
semiring} $\{+, 0, -, ?\}$, the rule-of-signs domain of abstract interpretation~\cite{cousot77}, retains
them: multiplication of signs is exact, and the only sum whose sign is not determined is $+$ plus $-$, which
is $?$. Taking signs as the scalars gives the \emph{qualitative derivatives} of qualitative
reasoning~\cite{dekleer84,kuipers86}, an entry reading ``increasing this input increases (or decreases, or
cannot affect) that output'', with $?$ appearing exactly at the potential cancellations: in the refund
example, the entry for the price becomes $?$. Collapsing $+$, $-$ and $?$ to $1$ recovers the Boolean
analysis.

\subsubsection{Linearised provenance}
\label{sec:approx-as-tangents:linearised-provenance}

Database provenance has its own semiring tradition: in the framework of
\citet{green07}, query results carry annotations from a commutative semiring, with the \emph{free}
commutative semiring $\mathbbm{N}[X]$ of \emph{provenance polynomials} playing a universal role: its
variables name the input tuples, multiplication records their joint use in a derivation, and addition
collects the alternative derivations. By freeness, assigning a value to each variable extends uniquely to a
semiring homomorphism out of $\mathbbm{N}[X]$, so provenance computed once as polynomials determines, by
evaluation, the annotation in every other semiring.

Provenance semirings can serve as the scalars of the present framework, but with a limitation: a
derivative is linear, so the framework captures \emph{linearised} provenance. Taking $S = \mathbbm{N}[X]$, a
Jacobian entry is a polynomial recording symbolically how an input position contributes to an output position.
Green's semantics assigns an output a polynomial $p$ whose degree records joint use; seeding each input
position with a unique variable, our forward analysis instead returns a polynomial linear in the position
variables, with the partial derivative $\partial p / \partial x_i$, evaluated at the run's input values, as
the coefficient of $x_i$. Joint use survives only in the coefficients: an input
used twice contributes $2x$ where the annotation would record $x^2$.

Translation along a homomorphism applies to the present framework as well: a semiring homomorphism acts
entrywise on Jacobians, and the action is functorial, since matrix composition is built from addition and
multiplication. The translation of interest for dependency tracking is ``zero-testing'' into the Booleans, and
whether that is a homomorphism depends on the source semiring. The counting semiring $\mathbbm{N}$, itself a
provenance semiring recording how many derivations use each input, is \emph{positive}: a sum is zero only when
both summands are, and a product only when a factor is; equivalently, zero-testing is a homomorphism, so
counting translates exactly to dependency tracking. The rationals are not positive, since a number and its
negation sum to zero, and zero-testing $\mathbbm{Q} \to \Two$ preserves multiplication exactly but addition
only laxly. It still acts entrywise on matrices, preserving identities and transposition, but composition only
up to over-approximation: composing Boolean Jacobians can record dependencies absent from the Boolean image of
their composite. The refund in \exref{introduction-example} is a case in point: the price genuinely does not
matter to the total, yet the composed dependencies say it might.

\begin{remark}
  \label{rem:jacobian-translation}
  For most of the semirings in this section the analysis is plainly different in kind from numeric
  differentiation: it reports use counts, symbolic polynomials or perturbation bounds rather than rates of
  change. Over the Booleans, however, one might hope for a shortcut: compute the rational Jacobian and then
  record, for each entry, only whether it is zero. Indeed, there is a potential gain in accuracy: composing
  Boolean Jacobians may record dependencies that the zero-tested rational Jacobian omits, never the reverse,
  so the shortcut sometimes gives a sharper answer at the given run. But it also answers a different
  question. Zero-testing certifies that the output is insensitive to the input at this run, and for a
  nonlinear operation that is a weaker property than the input being unnecessary: squaring has derivative
  zero at zero although its output is not constant. The Boolean analysis instead soundly approximates the
  inputs needed to recompute the output, the guarantee that \GPS relies on.
\end{remark}

\section{Models of Semiring Dependency} % \GPS doesn't use title caps
\label{sec:models-of-total-gps}

For basic differentiable programming, inputs and outputs are
  over Euclidean spaces where the tangent spaces are uniform and
  always isomorphic to the original space. It is then possible to, for
  first order programs at least, to identify a space and its
  tangent spaces. For semiring dependency analysis we cannot do this. It is
  not often the case that the types of data the programs operate will
  be the same as the semiings used for dependency tracking. To apply
  the semiring settings identified in \secref{approx-as-tangents}
  to a programming language we need an interpretation that attaches
  semimodules over our chosen semiring to types' interpretations, and
  linear maps to the interpretations of terms. Given the similarities
  noted in the introduction between provenance tracking and automatic
  differentiation, we follow V{\'a}k{\'a}r and collaborators
  \cite{vakar22,nunes2023} who used the \emph{Category of Families} to
  build denotational models of automatic differentiation over the
  reals. We start by recalling the definition of biproducts in
  categories enriched over commutative monoids, which will be the
  essential property of modules over semirings that we need.

  All of the results of this section have been mechanised in
  Agda %\footnote{Please refer to the supplementary material.}
  .

\subsection{CMon-Categories and Biproducts}
\label{sec:biproducts}

As we noted at the end of \secref{approx-as-tangents:semiring},
categories of semirings and semimodules have
\emph{biproducts}. Loosely stated, biproducts are objects that are
both products and coproducts. The concept can be defined in any
category, as shown by \citet{karvonen20}, but for our purposes it will
be more convenient to use the more concise definition in categories
enriched in commutative monoids:
\begin{definition}
  A category $\cat{C}$ is enriched in $\CMon$, the category of
  commutative monoids, if every homset $\cat{C}(X,Y)$ is a commutative
  monoid with $(+,0)$ and composition is bilinear:
  \begin{displaymath}
    f \comp \zero = 0 = \zero \comp f
  \end{displaymath}
  \begin{displaymath}
    (f + g) \comp h = (f \comp h) + (g \comp h) \qquad
    h \comp (f + g) = (h \comp f) + (h \comp g)
  \end{displaymath}
\end{definition}
In any $\CMon$-category we can define what it means to be the
biproduct of two objects:
\begin{definition}
  \label{def:biproducts}
  In a $\CMon$-category a biproduct is an object $X \biprod Y$
  together with morphisms

  \begin{center}
    \begin{tikzcd}
      X \arrow[r, "\biinj_X", shift left] &
      X \biprod Y \arrow[l, "\biproj_X", shift left] \arrow[r, "\biproj_Y"', shift right] &
      Y \arrow[l, "\biinj_Y"', shift right]
    \end{tikzcd}
  \end{center}

  \vspace{-1mm}
  \noindent satisfying

  \vspace{-5mm}
  \begin{minipage}[t]{0.45\textwidth}
    \begin{center}
      \begin{salign*}
        \biproj_X \comp \biinj_X &= \id_X \\
        \biproj_Y \comp \biinj_X &= \zero_{X,Y}
      \end{salign*}
    \end{center}
  \end{minipage}%
  \begin{minipage}[t]{0.45\textwidth}
    \begin{center}
      \begin{salign*}
        \biproj_Y \comp \biinj_Y &= \id_Y \\
        \biproj_X \comp \biinj_Y &= \zero_{Y,X}
      \end{salign*}
    \end{center}
  \end{minipage}

  \begin{salign*}
    (\biinj_X \comp \biproj_X) + (\biinj_Y \comp \biproj_Y) &= \id_{X \biprod Y}
  \end{salign*}
  A zero object is an object that is both initial and terminal.
\end{definition}
As the name suggests, biproducts in a category are both products and
coproducts:
\begin{proposition}
  \item
  \begin{enumerate}[leftmargin=\enummargin]
  \item A $\CMon$-category that has biproducts $X \oplus Y$ for all
    $X$ and $Y$ also has products and coproducts with
    $X \times Y = X + Y = X \oplus Y$.
  \item A $\CMon$-category with (co)products also has biproducts, and
    any initial or terminal object is a zero object.
  \end{enumerate}
\end{proposition}

\begin{example}
  The following are $\CMon$-enriched and have finite products, and
  hence biproducts:
  \begin{enumerate}[leftmargin=\enummargin]
  \item In $\FinVect$, morphisms are linear maps and so can be added
    and have a zero map. Finite products are given by Cartesian
    products of the underlying sets, with the vector operations
    defined pointwise.
  \item In $\SDSemiMod(S)$ and $\SemiMod(S)$ morphisms can be added
    and have a zero map. Biproducts are also given by Cartesian
    product of the underlying sets. The chosen self-duality is
    preserved because we always have $(M \oplus N)^*$ is a biproduct
    of $M^*$ and $N^*$.

  \end{enumerate}
\end{example}

\begin{remark}
  A $\CMon$-category having biproducts means that its finite products and coproducts coincide, but this does
  not necessarily extend to infinite products and coproducts. $\SemiMod(S)$ has infinite products and
  coproducts, but these do not coincide. For this reason, $\SDSemiMod(S)$ does not have infinite products
  because the dual of an infinite product is no longer uniquely isomorphic to an infinite product of duals.
  $\FinVect$ also lacks infinite products: the points of a product are tuples of points of the factors,
  so a product of countably many copies of $\RR$ would be infinite-dimensional. This will become
  relevant when we consider Cartesian closure in \secref{models:Cartesian-closure}.
\end{remark}

\begin{remark}
  \label{rem:biproduct-not-closed}
  Categories with zero objects cannot be Cartesian closed without
  being trivial in the sense of having exactly one morphism between
  every pair of objects because
  $\cat{C}(X, Y) \cong \cat{C}(1 \times X,Y) \cong \cat{C}(0 \times
  X,Y) \cong \cat{C}(0,X \to Y) \cong 1$. Consequently, we cannot
  apply the alternative construction of exponentials described in
  \remref{hermida-exponentials}.
\end{remark}

\subsection{The Category of Families Construction}
\label{sec:models-of-total-gps:fam}

As explained in \secref{approx-as-tangents:autodiff}, every point
in a manifold has an associated vector space of \emph{tangents} at
that point, and every smooth function has an associated function
mapping points to linear maps between tangent spaces. The
\emph{Category of Families} construction abstracts this situation,
forgetting any additional structure such as the topology or higher
derivatives and replacing the category of real finite-dimensional
vector spaces with an arbitrary category:

\begin{definition}
  Let $\cat{C}$ be a category. The \emph{Category of Families} over
  $\cat{C}$, $\Fam(\cat{C})$, has as objects pairs $(X, \partial X)$,
  where $X$ is a set and $\partial X : X \to \cat{C}$ is an
  $X$-indexed family of objects in $\cat{C}$. A morphism
  $f : (X, \partial X) \to (Y, \partial Y)$ consists of a pair of a
  function $f : X \to Y$ and a family of morphisms of $\cat{C}$,
  $\partial f : \Pi_{x : X}.\,\cat{C}(\partial X(x), \partial Y(f\,x))$.
\end{definition}

The reason for choosing the $\Fam$ construction is that composition in
this category is an abstract version of the chain rule that we have
seen in \remref{chain-rule}. Composition $f \circ g $ of morphisms
$f : (Y, \partial Y) \to (Z, \partial Z)$ and
$g : (X, \partial X) \to (Y, \partial Y)$ in this category is given by
normal function composition on the set components, and
$\partial (f \circ g)(x) = \partial f(f, x) \circ \partial g(x)$,
where the latter composition is in $\cat{C}$.

The fact that morphisms in $\Fam(C)$ compose according to a chain rule
means that the categories we considered in \secref{approx-as-tangents}
embed into $\Fam(C)$ for appropriate $C$. If we let $\FinVect$ be the
category of finite dimensional real vector spaces and linear maps,
then:

\begin{proposition}
  \label{prop:embed-manifolds}
  There is a faithful functor $\Man \to \Fam(\FinVect)$ that sends a
  manifold $M$ to $(M, \lambda x. T_x(M))$, and each smooth function
  $f$ to $(f, f_*)$, the pair of $f$ and its forward derivative.
\end{proposition}

A similar result is given by \citet{cruttwell2022}, where Euclidean
spaces $\RR^n$ and smooth functions are embedded into a category of
lenses (the ``simply typed'' version of the $\Fam$ construction). As
in \citet{vakar22}, the idea is to formally separate functions on
points and their forward/reverse tangent maps for the purposes of
implementation of automatic differentiation. In the case of smooth
maps, this process throws away information on higher derivatives by
turning smooth maps into pairs of plain functions and linear
functions.

We will use the categories $\Fam(\SDSemiMod(S))$ as our model
  of semiring-tracking data provenance, analogous to V{\'a}k{\'a}r
  \emph{et al.}'s use of $\Fam(\FinVect)$. The self-duality of objects
  in $\SDSemiMod(S)$ means that we can take transposes of tangent maps
  to track reverse provenance. To interpret expressive languages in
  this setting, we need to ensure it has sufficient categorical
  structure. Unfortunately, it is the case that $\Fam(\SDSemiMod(S))$
  is not Cartesian closed, and so cannot support an interpretation of
  higher order functions. This motivates the use of the more liberal
  category $\Fam(\SemiMod(S))$ to interpret higher order programs.

\subsection{Categorical Properties of $\Fam(\cat{C})$}

\subsubsection{Coproducts and Products}
\label{sec:models-of-total-gps:coproducts-and-products}

The categories $\Fam(\cat{C})$ are the free coproduct completions of
categories $\cat{C}$, so they have all coproducts:

\begin{proposition}
  For any $\cat{C}$, $\Fam(\cat{C})$ has all coproducts, which can be
  given on objects by:
  \begin{displaymath}
    \coprod_i (X_i, \partial X_i) = (\coprod_i X_i, \lambda (i, x_i).\, \partial X_i(x))
  \end{displaymath}
  Coproducts in $\Fam(\cat{C})$ are \emph{extensive}
  \cite[Proposition 2.4]{carboni-lack-walters93}.
\end{proposition}

For $\Fam(\cat{C})$ to have finite products, we need $\cat{C}$ to have
finite products:

\begin{proposition}
  If $\cat{C}$ has finite products, then so does $\Fam(\cat{C})$. On
  objects, binary products can be defined by:
  \begin{displaymath}
    (X, \partial X) \times (Y, \partial Y) = (X \times Y, \lambda (x, y). \partial X(x) \times \partial Y(y))
  \end{displaymath}
  Since $\Fam(\cat{C})$ is extensive, products and coproducts
  distribute.
\end{proposition}

Using the infinitary coproducts and finite products, we can construct
a wide range of other useful semantic models of datatypes in
$\Fam(\cat{C})$. For example, lists can be constructed as a coproduct
\begin{equation}
  \label{eqn:lists}
  \List(X) = \coprod_{n \in \mathbb{N}} X^n
\end{equation}
where $X^0 = 1$ (the terminal object) and $X^{n+1} = X \times X^n$.

Our category of interest, $\Fam(\SDSemiMod(S))$ has coproducts
  and finite products, because $\SDSemiMod(S)$ has
  (bi)products. Similarly for $\Fam(\SemiMod(S))$.

A direct consequence of our chosen construction of products and
coproducts in $\Fam(C)$ is:

\begin{lemma}
  \label{lem:pi1-preserve-products-and-coproducts}
  The first projection functor $\pi_1 : \Fam(C) \to \Set$ preserves
  all products and coproducts.
\end{lemma}

\subsubsection{Cartesian Closure}
\label{sec:models:Cartesian-closure}

For Cartesian closure of the categories $\Fam(C)$ that we are
interested in, we rely on the following theorem of \citet{nunes2023},
specialised from their setting with the general Grothendieck
construction to $\Fam(\cat{C})$.

\begin{theorem}[\cite{nunes2023}]
  \label{thm:fam-closed}
  \AGDA.  If $\cat{C}$ has biproducts (\defref{biproducts}) and all
  products, then $\Fam(\cat{C})$ is Cartesian closed\footnote{More
    precisely, if $\cat{C}$ has coproducts then we have a monoidal
    product on $\Fam(\cat{C})$ which is closed by this
    construction. When these coproducts are in fact biproducts, we get
    Cartesian closure.}. On objects, the internal hom can be given by:
  \begin{displaymath}
    (X, \partial X) \to (Y, \partial Y) = (\Pi_{x : X}. \Sigma_{y : Y}. \cat{C}(\partial X(x), \partial Y(y)), \lambda f. \Pi_{x : X}. \partial Y(\pi_1(f\, x)))
  \end{displaymath}
\end{theorem}

The $\Set$-component of $(X, \partial X) \to (Y, \partial Y)$ consists
of exactly the morphisms of $\Fam(\cat{C})$, rephrased into a single
object. When $\cat{C} = \FinVect$ or $\SemiMod(S)$, these are functions with an
associated linear map at every point. A tangent to a function is then defined to be a mapping from
points in the domain to tangents in the codomain along the function.

The category $\SemiMod(S)$ satisfies the hypotheses of
\thmref{fam-closed}, so:
\begin{corollary}
  $\SemiMod(S)$ is Cartesian closed and has all coproducts.
\end{corollary}
Unfortunately, neither $\SDSemiMod(S)$ nor
$\FinVect$ satisfy the hypotheses of this theorem, because as we noted
above, neither of them have infinite products. So we are forced to use
  $\Fam(\SemiMod(S))$ to interpret higher order functions, which
  appears to forego the ability to obtain reverse derivatives by
  transpose of linear maps afforded by $\Fam(\SDSemiMod(S))$. However,
  $\Fam(\SDSemiMod(S))$ fully embeds into $\Fam(\SemiMod(S))$, so we
  can read back interpretations of first-order programs and compute
  their transpose derivatives.

\begin{proposition}
  \label{prop:ho-embedding}
  There is a functor $H : \Fam(\SDSemiMod(S)) \to \Fam(\SemiMod(S))$
  defined on objects as
  $H(X, \partial X) = (X, \lambda x. U(\partial X(x)))$, where
  $U : \SDSemiMod(S) \to \SemiMod(S)$ is the functor that forgets the
  self-dual data. The functor $H$ preserves products (because $U$
  does) and all coproducts. It is also full and faithful.
\end{proposition}

As we will see in \secref{definability}, this is no problem
  for programs of first order type.

\begin{remark}
  \label{rem:hermida-exponentials}
  There is another construction of internal homs on $\Fam(\cat{C})$
  arising from the use of fibrations for categorical logical
  relations, due to \citet[Corollary 4.12]{hermida99}. If we assume
  that $\cat{C}$ is itself Cartesian closed and has all products, then
  we could construct an internal hom as:
  \begin{displaymath}
    (X, \partial X) \to (Y, \partial Y) = (X \to Y, \lambda f. \Pi_{x : X}.\,\partial X(x) \to \partial Y(f\,x))
  \end{displaymath}
  However, for the purposes of modelling differentiable programs, this
  is fatally flawed in that neither $\SemiMod(S)$ nor $\FinVect$ are
  Cartesian closed, and there is no way of making them so without
  losing the property of being able to conjunct or add tangents
  (\remref{biproduct-not-closed}).
\end{remark}

\section{Higher-Order Language}
\label{sec:language}

To demonstrate semiring provenance tracking via automatic differentiation for higher-order programs, we define a simple total functional
programming language, extending the simply-typed lambda calculus. The language is parameterised by a signature
$\Sigma = (\PrimTy, \PrimOp)$ consisting of a set $\PrimTy$ of base types $\rho$ and a family of sets
$\PrimOp^\rho_{\rho_1,\ldots,\rho_n}$ of primitive operations $\phi$ of arity $n$ over those base types.

\subsection{Syntax}
\label{sec:language:syntax}

The syntax is defined in \figref{syntax}. Types includes base types $\rho$ drawn from $\PrimTy$, along with
standard type formers for sums, products, functions and lists% , plus a type constructor $\tyLift$ to allow
% approximation points to be added explicitly, as discussed in \secref{models-of-total-gps}
. Terms include
variables, the usual introduction and elimination forms% , a monadic return and bind for lifted terms
, and
primitive operations $\phi$.

The language is intentionally minimal: it excludes general recursion, and general inductive or coinductive
types, which we will consider in future work (\secref{conclusion}). Typing judgments for terms are standard
and shown in \figref{typing}, with the usual rules for products, sums, functions, and lists.

\begin{figure}
  \begin{subfigure}[t]{0.48\linewidth}
  \small
  \[
  \begin{array}{lllll}
    & \textit{Types}
    \\
    &
    \sigma, \tau
    & ::= &
    \rho
    &
    \text{primitive type}
    \\
    && \mid &
    \sigma \tySum \tau
    &
    \text{sum}
    \\
    && \mid &
    \tyUnit
    &
    \text{unit}
    \\
    && \mid &
    \sigma \tyProd \tau
    &
    \text{product}
    \\
    && \mid &
    \sigma \tyFun \tau
    &
    \text{function}
    \\
    && \mid &
    \tyList\;\tau
    &
    \text{list}
    % \\
    % && \mid &
    % \tyLift\;\tau
    % &
    % \text{lifting}
  \end{array}
  \]
  \end{subfigure}%
  \begin{subfigure}[t]{0.48\linewidth}
  \small
  \[
  \begin{array}{lllll}
    & \textit{Terms}
    \\
    &
    t, s
    & ::= &
    x
    &
    \text{variable}
    \\
    && \mid &
    \phi(\vec t)
    &
    \text{primitive op}
    \\
    && \mid &
    \tmInl{t} \mid \tmInr{t}
    &
    \text{injection}
    \\
    && \mid &
    \tmCase{s}{x}{t_1}{y}{t_2}
    &
    \text{case}
    \\
    && \mid &
    \tmUnit
    &
    \text{unit}
    \\
    && \mid &
    \tmPair{s}{t}
    &
    \text{pair}
    \\
    && \mid &
    \tmFst{t} \mid \tmSnd{t}
    &
    \text{projection}
    \\
    && \mid &
    \tmFun{x}{t}
    &
    \text{function}
    \\
    && \mid &
    \tmApp{s}{t}
    &
    \text{application}
    \\
    && \mid &
    \tmNil \mid \tmCons{s}{t}
    &
    \text{nil \& cons}
    % \\
    % && \mid &
    % \tmCons{s}{t}
    % &
    % \text{cons}
    \\
    && \mid &
    \tmFoldList{s_1}{s_2}{t}
    &
    \text{fold}
    % \\
    % && \mid &
    % \tmReturn{t}
    % &
    % \text{return}
    % \\
    % && \mid &
    % \tmBind{s}{t}
    % &
    % \text{bind}
  \end{array}
  \]
  \end{subfigure}
  \caption{Syntax of types and terms}
  \label{fig:syntax}
\end{figure}

\begin{figure}
  \begin{mathpar}
    \small
    \inferrule*
    {
      x : \tau \in \Gamma
    }
    {
      \Gamma \vdash x: \tau
    }
    \and
    \inferrule*
    {
      \phi \in \PrimOp^\rho_{\rho_1, \ldots, \rho_n}
      \\
      \Gamma \vdash t_i: \rho_i
      \quad
      (\forall i \in \{1..n\})
    }
    {
      \Gamma \vdash \phi(t_1, \ldots, t_n): \rho
    }
    \and
    \inferrule*
    {
      \Gamma \vdash t : \sigma
    }
    {
      \Gamma \vdash \tmInl{t}: \sigma \tySum \tau
    }
    \and
    \inferrule*
    {
      \Gamma \vdash t : \tau
    }
    {
      \Gamma \vdash \tmInr{t}: \sigma \tySum \tau
    }
    \and
    \inferrule*
    {
      \Gamma \vdash s : \sigma \tySum \tau
      \\
      \Gamma, x: \sigma \vdash t_1 : \tau'
      \\
      \Gamma, y : \tau \vdash t_2 : \tau'
    }
    {
      \Gamma \vdash \tmCase{s}{x}{t_1}{y}{t_2}: \tau'
    }
    \and
    \inferrule*
    {
      \strut
    }
    {
      \Gamma \vdash \tmUnit : \tyUnit
    }
    \and
    \inferrule*
    {
      \Gamma \vdash s : \sigma
      \\
      \Gamma \vdash t : \tau
    }
    {
      \Gamma \vdash \tmPair{s}{t}: \sigma \tyProd \tau
    }
    \and
    \inferrule*
    {
      \Gamma \vdash t : \sigma \tyProd \tau
    }
    {
      \Gamma \vdash \tmFst{t}: \sigma
    }
    \and
    \inferrule*
    {
      \Gamma \vdash t : \sigma \tyProd \tau
    }
    {
      \Gamma \vdash \tmSnd{t}: \tau
    }
    \and
    \inferrule*
    {
      \Gamma, x: \sigma \vdash t : \tau
    }
    {
      \Gamma \vdash \tmFun{x}{t}: \sigma \tyFun \tau
    }
    \and
    \inferrule*
    {
      \Gamma \vdash s: \sigma \tyFun \tau
      \\
      \Gamma \vdash t : \sigma
    }
    {
      \Gamma \vdash \tmApp{s}{t}: \tau
    }
%    \and
%    \inferrule*
%    {
%      \Gamma \vdash t : \subst{\tau}{\mu \alpha.\tau}{\alpha}
%    }
%    {
%      \Gamma \vdash \tmRoll{t}: \mu\alpha.\tau
%    }
%    \and
%    \inferrule*
%    {
%      \Gamma \vdash s : \subst{\sigma}{\tau}{\alpha} \tyFun \tau
%      \\
%      \Gamma \vdash t : \mu\alpha.\sigma
%    }
%    {
%      \Gamma \vdash \tmFold{s}{t} : \tau
%    }
    \and
    \inferrule*
    {
      \strut
    }
    {
      \Gamma \vdash \tmNil : \tyList\;\tau
    }
    \and
    \inferrule*
    {
      \Gamma \vdash s: \tau
      \\
      \Gamma \vdash t: \tyList\;\tau
    }
    {
      \Gamma \vdash \tmCons{s}{t} : \tyList\;\tau
    }
    \and
    \inferrule*
    {
      \Gamma \vdash s_1 : \tau
      \\
      \Gamma, x: \sigma, y: \tau \vdash s_2 : \tau
      \\
      \Gamma \vdash t : \tyList\;\sigma
    }
    {
      \Gamma \vdash \tmFoldList{s_1}{s_2}{t} : \tau
    }
    % \and
    % \inferrule*
    % {
    %   \Gamma \vdash t : \tau
    % }
    % {
    %   \Gamma \vdash \tmReturn{t} : \tyLift\;\tau
    % }
    % \and
    % \inferrule*
    % {
    %   \Gamma \vdash s : \tyLift\;\sigma
    %   \\
    %   \Gamma \vdash t : \sigma \tyFun \tyLift\;\tau
    % }
    % {
    %   \Gamma \vdash \tmBind{s}{t} : \tyLift\;\tau
    % }
  \end{mathpar}
\caption{Well-typed terms over a signature $\Sigma$}
\label{fig:typing}
\end{figure}

\subsection{Semantics}
\label{sec:language:semantics}

\begin{figure}
  \begin{subfigure}[t]{0.55\linewidth}
    \begin{minipage}{0.55\linewidth}
    \small
    \begin{align*}
      \sem{\rho} &= \sem{\rho}_{\PrimTy}
      \\
      \sem{\sigma \tySum \tau} &= \textstyle {\sem{\sigma}} + {\sem{\tau}}
      \\
      \sem{\tyUnit} &= 1
    \end{align*}
  \end{minipage}
    \begin{minipage}{0.4\linewidth}
    \small
    \begin{align*}
      \sem{\sigma \tyProd \tau} &= \sem{\sigma} \times \sem{\tau}
      \\
      \sem{\sigma \tyFun \tau} &= \internalHom{\sem{\sigma}}{\sem{\tau}}
      \\
      \sem{\tyList\;\tau} &= \List(\sem{\tau})
      % \\
      % \sem{\tyLift\;\tau} &= \Lift(\sem{\tau})
    \end{align*}
    \end{minipage}
    \caption{Interpretation of Types}
    \label{fig:semantics:types}
  \end{subfigure}
  \begin{subfigure}[t]{0.4\linewidth}
    \small
    \begin{align*}
      \sem{\emptyCxt} &= 1
      \\
      \sem{\Gamma, x: \tau} &= \sem{\Gamma} \times \sem{\tau}
    \end{align*}
    \caption{Interpretation of Contexts}
    \label{fig:semantics:contexts}
\end{subfigure}
\begin{subfigure}{0.8\linewidth}
  \begin{minipage}{0.5\linewidth}
  \small
  \begin{align*}
  \sem{x_i} &= \pi_i
  \\
  \sem{\phi(t_1, \ldots, t_n)}
  &=
  \sem{\phi}_{\Op} \comp \prodM{\sem{t_1}}{\ldots, \sem{t_n}}
  \\
  \sem{\tmInl{t}} &= \mathsf{inj}_1 \comp \sem{t}
  \\
  \sem{\tmInr{t}} &= \mathsf{inj}_2 \comp \sem{t}
  \\
  \sem{\tmCase{s}{x}{t_1}{y}{t_2}} &= \coprodM{\sem{t_1}}{\sem{t_2}} \comp \prodM{\id}{\sem{s}}
  \\
  \sem{\tmUnit} &=\;!_{\sem{\Gamma}}
  \\
    \sem{\tmPair{s}{t}} &= \prodM{\sem{s}}{\sem{t}}
  \end{align*}
\end{minipage}
\begin{minipage}{0.5\linewidth}
  \small
  \begin{align*}
  \sem{\tmFst{t}} &= \pi_1 \comp \sem{t}
  \\
  \sem{\tmSnd{t}} &= \pi_2 \comp \sem{t}
  \\
  \sem{\tmFun{x}{t}} &= \lambda(\sem{t})
  \\
  \sem{\tmApp{s}{t}} &= \eval \comp \prodM{\sem{s}}{\sem{t}}
  \\
  \sem{\tmNil} &= \nil \comp {!_{\sem{\Gamma}}}
  \\
  \sem{\tmCons{s}{t}} &= \cons \comp \prodM{\sem{s}}{\sem{t}}
  \\
  \sem{\tmFoldList{t_1}{t_2}{s}} &= \fold(\sem{t_1},\sem{t_2}) \comp \prodM{\id}{\sem{s}}
  % \\
  % \sem{\tmReturn{t}} &= \eta_{\sem{\sigma}} \comp \sem{t}
  % \tag*{($\Gamma \vdash t: \sigma$)}
  % \\
  % \sem{\tmBind{s}{t}} &= \mu_{\sem{\tau}} \comp \Lift(\sem{t}) \comp \mathsf{st}_{\sem{\Gamma},\sem{\sigma}} \comp \prodM{\id}{\sem{s}}
  % \tag*{($\Gamma \vdash t: \sigma \to \tyLift\;\tau$)}
  \end{align*}
  \end{minipage}
  \caption{Terms as morphisms}
  \label{fig:semantics:terms}
\end{subfigure}
\caption{Interpretation of types, contexts and terms} % in category $\Sem$
\end{figure}

An interpretation of a signature $\Sigma = (\PrimTy, \PrimOp)$ can be given in any category $\cat{C}$ with
finite products, and assigns to each base type $\rho \in \PrimTy$ an object $\sem{\rho}_{\PrimTy}$ in
$\cat{C}$, and to each primitive operation $\phi \in \PrimOp^\rho_{\rho_1,\ldots,\rho_n}$, a morphism
$\sem{\phi}_{\Op}: \sem{\rho_1}_{\PrimTy} \times \ldots \times \sem{\rho_n}_{\PrimTy} \to
\sem{\rho}_{\PrimTy}$.

Assuming that $\cat{C}$ is bicartesian closed and has a list object (\eqnref{lists}), then we can extend an
interpretation of a signature $\Sigma$ to an interpretation of the whole language over
$\Sigma$. \figref{semantics:types} and \figref{semantics:contexts} define the interpretation of types and
contexts as objects of $\cat{C}$ respectively. Terms are interpreted as morphisms between the interpretations
of the context and type, as defined in \figref{semantics:terms}. We have used the notations $\pi_i$ for
projections, $\prodM{f}{g}$ for pairing, $\coprodM{f}{g}$ for (parameterised) copairing, $!_X$ for morphisms
to the terminal object, and $\lambda$ and $\eval$ for currying and evaluation for exponentials.

We have another interpretation $\sem{-}_{\mathit{fo}}$ of the
first-order types (those constructed from primitive types, sums, unit
and products) in any bicartesian category. Such interpretations are
preserved by finite coproduct and coproduct preserving functors:

\begin{lemma}\label{lem:first-order-agreement-types}
  If $\cat{C}$ and $\cat{D}$ are bicartesian and bicartesian closed categories with interpretations of the
  signature $\Sigma$, $F : \cat{C} \to \cat{D}$ is a bicartesian functor, and
  $F(\sem{\rho}_{\PrimTy}) \cong \sem{\rho}_{\PrimTy}$ for all $\rho$, then for all first-order types $\tau$,
  $F(\sem{\tau}_{\mathit{fo}}) \cong \sem{\tau}$, and similar for contexts only containing first-order types.
\end{lemma}

\subsubsection{Interpretation in the Semimodule Model}
\label{sec:language:gps-interpretation}

Given the above, we can now interpret the language in any of the bicartesian closed categories with list
objects we constructed in \secref{models-of-total-gps}. Specifically, we assume that we have an interpretation
of our chosen signature in $\Fam(\SDSemiMod(S))$. Each base type is
assigned a self-dual approximation object, and the first-order type formers preserve self-duality, so every
first-order type denotes a self-dual semimodule. Signatures are
first-order, so it does not matter that $\Fam(\SDSemiMod(S))$ is not closed. Any
such interpretation can be transported to $\Fam(\SemiMod(S))$ along the functor $H$ from \propref{ho-embedding} because it preserves finite products. We
can then interpret the whole language in $\Fam(\SemiMod(S))$.

Interpreting a program $\Gamma \vdash t : \tau$ yields a morphism of $\Fam(\SemiMod(S))$, the
underlying function together with its forward derivative. In contrast to many other bidirectional program
analysis techniques, a single morphism is enough, because a linear map has a transpose for free. At first-order types, \lemref{first-order-agreement-types} identifies the
interpretation with that of the first-order model $\Fam(\SDSemiMod(S))$, whose objects are self-dual, so the
forward derivative is a morphism of $\Fam(\SDSemiMod(S))$ whose conjugate $\conj{f}$ is the backward map
between the same semimodules. 

% Given a model of $\Sigma$ in $\Fam(\LatGal)$, the category $\Fam(\MeetSLat \times \JoinSLat^\op)$ can
% implement $\Sem$. It is bicartesian closed by \corref{mslat-jslat-bcc}, and can model $\tyList$ using infinite
% coproducts (\secref{models-of-total-gps:coproducts-and-products}). The embedding $\HoEmbed: \Fam(\LatGal) \to
% \Fam(\MeetSLat \times \JoinSLat^\op)$ from \propref{ho-embedding} preserves finite products, and so can
% transport each $\sem{\rho}_{\PrimTy}$ and $\sem{\phi}_{\Op}$ into the higher-order setting to provide a model
% of $\Sigma$. Finally the lifting monad in $\Fam(\LatGal)$ (\secref{first-order:lifting}) is also preserved by
% $H$. \roly{don't think we actually show this}

% The correctness of this interpretation is considered in the next section.

\section{Examples}
\label{sec:examples}

We now run the interpretation on a range of example programs, over various semirings, with the aim
of illustrating the different flavours of analysis the framework supports. The examples are kept
small so that each can be mechanised in the accompanying Agda development, with the programs and their
derivatives run inside the type checker and the results checked against the expected values by decidable
equality.

Throughout, we work over the signature $\Sigma_{\mathrm{ex}}$, parameterised by a set
of numbers with a distinguished zero, which the examples instantiate with the rationals or the integers. It has sorts
$\mathsf{num}$ and $\mathsf{label}$; operations $\mathrm{add}, \mathrm{mult} : \mathsf{num} \times
\mathsf{num} \to \mathsf{num}$, a literal for each number, and a literal for each label; and an equality
relation on labels, used by the comprehension syntax. In the intended interpretations, the zero literal is a
unit for $\mathrm{add}$. Each instance must also interpret each operation of the signature, supplying the
function it computes on values together with, at each input value, a linear map between the approximation
objects, a semimodule morphism giving the operation's intended derivative at that value.

\subsection{Perturbation Bounds}
\label{sec:examples:perturbation-bounds}

The first two instances demonstrate the perturbation-bound reading of
\secref{approx-as-tangents:perturbation-bounds}: the analysis reports how much each output can change, rather
than whether it can. The two presentations, absolute (\secref{examples:absolute-bounds}) and relative
(\secref{examples:relative-bounds}), run on the database $\mathit{db} = [(\mathsf{a}, 3),
(\mathsf{b}, 1), (\mathsf{a}, -3)]$ of \exref{introduction-example}. In the absolute presentation the
approximation object for a number is the free semimodule $S^2$, one bound for each direction of change; a
relative bound constrains both directions at once, so the relative presentation uses a single scalar.

\subsubsection{Absolute bounds}
\label{sec:examples:absolute-bounds}

At the min-plus semiring, a pair of scalars $(d, u)$ records that a value may fall by at most $d$ and
rise by at most $u$, so that a number $x$ is known to lie in the interval $[x - d, x + u]$; the pair
$(\infty, \infty)$ carries no information. Multiplication's derivative carries none either, since no
translation bounds a scaling (\secref{approx-as-tangents:perturbation-bounds}), so we run a query that
selects and sums without multiplying: $\mathrm{query}\,l\,\mathit{db} = \mathrm{sum}\,[ q \mid (l', q)
\leftarrow \mathit{db}, l \equiv l' ]$, with $\mathrm{query}\,\mathsf{a}\,\mathit{db} = 0$. Addition's
derivative entries are translations by $0$, and selection's entries are the semiring's $0$ and $1$, as
before. The backward derivative at this run is a linear map $\partial(\mathrm{query}\,l\,\mathit{db})_r :
S^2 \multimap S^2 \times S^2 \times S^2$, one factor for each number in the database; the frozen label
positions and list terminator contribute trivial factors, omitted here, and because the input is a list, the
type depends on the input's shape.

Forwards, with the first $\mathsf{a}$-number known to within bounds $(\frac{1}{2}, 0)$ and the second
to within $(\frac{1}{5}, \frac{1}{2})$:
\begin{displaymath}
  \partial(\mathrm{query}\,\mathsf{a}\,\mathit{db})_f([(\mathsf{a}, (\tfrac{1}{2}, 0)), (\mathsf{b}, (\infty, \infty)), (\mathsf{a}, (\tfrac{1}{5}, \tfrac{1}{2}))]) =
  (\tfrac{1}{5}, 0)
\end{displaymath}
\noindent Bounds for the same position combine by $\min$, reflecting the reading in which one input
changes at a time. Backwards, asking for the output to lie within $[-\frac{1}{10}, \frac{1}{10}]$, that is,
bounds $(\frac{1}{10}, \frac{1}{10})$:
\begin{displaymath}
  \partial(\mathrm{query}\,\mathsf{a}\,\mathit{db})_r(\tfrac{1}{10}, \tfrac{1}{10}) =
  [(\mathsf{a}, (\tfrac{1}{10}, \tfrac{1}{10})), (\mathsf{b}, (\infty, \infty)), (\mathsf{a}, (\tfrac{1}{10}, \tfrac{1}{10}))]
\end{displaymath}
\noindent Changing either $\mathsf{a}$-number alone by at most $\frac{1}{10}$ keeps the output within
its interval, while the $\mathsf{b}$-number is unconstrained. The backward run here can only route the
demand to the used positions, since every entry is $0$ or $\infty$; it would become informative given
operations that loosen a bound by a fixed finite amount, such as rounding, whose entry is the translation
$1$ ($\delta$ in, $\delta + 1$ out), or several outputs sharing an input, whose demands combine by
$\min$.

\subsubsection{Relative bounds}
\label{sec:examples:relative-bounds}

Relative bounds pass through multiplication, so this instance can run the full weighted-sum query
of \exref{introduction-example}, prices included, on the same database. We instantiate at
the min-times semiring of \secref{approx-as-tangents:perturbation-bounds}: multiplication's derivative has
unit entries, and addition's entries are condition numbers, infinite at exact cancellation. In the run of
$\mathrm{total}\,\mathsf{b}$, multiplication behaves well: a $10\%$ bound on the price of $\mathsf{b}$
reaches the output intact, and likewise a bound on the selected quantity:
\begin{displaymath}
  \partial(\mathrm{total}\,\mathsf{b}\,\mathit{db}\,\mathit{price})_f([(\mathsf{a}, \infty), (\mathsf{b}, \infty), (\mathsf{a}, \infty)], \langle \infty, \tfrac{1}{10} \rangle) = \tfrac{1}{10}
\end{displaymath}
\noindent Backwards, an output bound of $10\%$ constrains the selected quantity and its price and
nothing else:
\begin{displaymath}
  \partial(\mathrm{total}\,\mathsf{b}\,\mathit{db}\,\mathit{price})_r(\tfrac{1}{10}) =
  ([(\mathsf{a}, \infty), (\mathsf{b}, \tfrac{1}{10}), (\mathsf{a}, \infty)], \langle \infty, \tfrac{1}{10} \rangle)
\end{displaymath}
\noindent The forward run of $\mathrm{total}\,\mathsf{a}$ meets the cancellation: its two contributions,
$2 \cdot 3$ and $2 \cdot (-3)$, sum to zero, the amplification factor of the final addition is infinite, and
no relative bound on either $\mathsf{a}$-quantity, however tight, yields a bound on the output:
\begin{displaymath}
  \partial(\mathrm{total}\,\mathsf{a}\,\mathit{db}\,\mathit{price})_f([(\mathsf{a}, \tfrac{1}{10}), (\mathsf{b}, \infty), (\mathsf{a}, \infty)], \langle \infty, \infty \rangle) = \infty
\end{displaymath}

\subsection{Linearised provenance}
\label{sec:examples:counting}

At the free semiring $\mathbbm{N}[X]$ the analysis computes the linearised provenance of
\secref{approx-as-tangents:linearised-provenance}. Each instance must specify the derivative entries of the
primitive operations; here every entry is $1$, so a Jacobian entry of a run counts the paths from an input
position to an output position. On the cancellation run of the weighted-sum query
(\exref{introduction-example}), seeding each input position with a unique variable:
\begin{displaymath}
  \partial(\mathrm{total}\,\mathsf{a}\,\mathit{db}\,\mathit{price})_f([(\mathsf{a}, x_0), (\mathsf{b}, x_1), (\mathsf{a}, x_2)], \langle x_3, x_4 \rangle) = x_0 + x_2 + 2 x_3
\end{displaymath}
\noindent The output polynomial records each selected quantity once and the queried price twice, once
per selected row. Backwards, seeding the output with a variable $y$ distributes it by usage:
\begin{displaymath}
  \partial(\mathrm{total}\,\mathsf{a}\,\mathit{db}\,\mathit{price})_r(y) =
  ([(\mathsf{a}, y), (\mathsf{b}, 0), (\mathsf{a}, y)], \langle 2 y, 0 \rangle)
\end{displaymath}
\noindent Other analyses now come by evaluation. Setting every variable to $1$ specialises to the
counting semiring $(\mathbbm{N}, +, 0, \cdot, 1)$: the forward polynomial evaluates to $4$, the number of
paths from perturbable inputs to the output, and the backward polynomials evaluate to per-position use counts,
$2$ at the queried price. The count survives the cancellation that zeroes the rational derivative, because
$\mathbbm{N}$ is positive and counts cannot cancel; zero-testing then sends the count $2$ to the Boolean $1$
exactly, as positivity guarantees, whereas the same translation applied to the rational entry loses the
dependency. The three analyses answer three different questions at the same position: the rational Jacobian
reports sensitivity ($0$), the Boolean Jacobian possible dependence ($1$), and the counting Jacobian usage
($2$).

\subsection{Signed Saliency}
\label{sec:examples:saliency}

Consider a program that condenses many numeric inputs into a single score, such as a model whose
predictions we want to explain. Ours scores a $3 \times 3$ grid of numbers against a fixed pattern,
like a small image filter applied to one patch of a larger image. The pattern weights the
centre positively and the corners negatively, and includes two products:
\begin{displaymath}
  \mathrm{score}
  \begin{pmatrix} x_1 & x_2 & x_3 \\ x_4 & x_5 & x_6 \\ x_7 & x_8 & x_9 \end{pmatrix}
  = x_5 - (x_1 + x_3 + x_7 + x_9) + x_4 x_6 - x_5 x_2
\end{displaymath}
\noindent A \emph{saliency map} records the influence of each input on the score; gradient-based
explanation methods obtain one from the model's derivative~\cite{selvaraju20}. Here we compute a
\emph{signed} saliency map over the sign semiring of \secref{approx-as-tangents:qualitative}, taking the
integers as the numbers: for each cell, whether increasing it raises or lowers the score. The signs
could be read off the rational derivative, but each verdict would then hold only at the particular
magnitudes of the run. Computing with signs as the derivatives instead gives verdicts that rest on signs
alone, so a definite entry holds under any sign-preserving change of magnitude, and $?$ marks an input
whose direction of influence signs alone are insufficient to determine. For addition, the derivative
entries are $+$; for multiplication, the entry for each argument is the sign of the other.

Running on the grid with rows $(1, 2, 1)$, $(3, 5, 4)$ and $(1, 7, 1)$ gives the score $3$. Seeding
the backward derivative at this run with $+$ distributes it to the inputs:
\begin{displaymath}
  \partial(\mathrm{score})_r(+) =
  \begin{pmatrix} - & - & - \\ + & ? & + \\ - & 0 & - \end{pmatrix}
\end{displaymath}
\noindent The four corners lower the score and the two mid-edge cells raise it, via the product
$x_4 x_6$; the bottom-middle cell, which the score ignores, gets $0$. At the centre the linear term and
the product $-x_5 x_2$ pull in opposite directions, so the entry is $?$. The rational derivative would instead report a negative
entry for the centre, a verdict that depends on the magnitude $x_2 = 2$ rather than on its sign.

\subsection{Cognacy Queries}
\label{sec:examples:linked}

The final example introduces no new semiring: it runs over the Booleans, and shows instead how
composing the two derivatives of a single run recovers an existing dependency analysis. When outputs overlap
in their dependencies, the composite of the backward and forward derivatives sends a selection of outputs to
the outputs related to it by a shared dependency: the \emph{cognacy} analysis of linked
visualisations~\cite{perera22,bond25,psallidas18}. The \emph{moving average} is the standard example of such
overlap. With window two and four inputs,
\begin{displaymath}
  \mathrm{mavg}\,(x_1, x_2, x_3, x_4) =
  (\tfrac{1}{2} (x_1 + x_2),\; \tfrac{1}{2} (x_2 + x_3),\; \tfrac{1}{2} (x_3 + x_4))
\end{displaymath}
\noindent Each adjacent pair of outputs shares an input, and the
first and last outputs share none. The cognacy analysis is a \emph{relation} on the outputs, derived from the
map's input-output relation, so it is useful to present the matrices involved. Running over the Booleans, the
forward derivative and its composite with the backward derivative are
\begin{displaymath}
  \partial_{\Two}(\mathrm{mavg}) =
  \begin{pmatrix} 1 & 1 & 0 & 0 \\ 0 & 1 & 1 & 0 \\ 0 & 0 & 1 & 1 \end{pmatrix}
  \qquad
  \partial_{\Two}(\mathrm{mavg}) \comp \transpose{\partial_{\Two}(\mathrm{mavg})} =
  \begin{pmatrix} 1 & 1 & 0 \\ 1 & 1 & 1 \\ 0 & 1 & 1 \end{pmatrix}
\end{displaymath}
The composite relates each output to the outputs it shares an input with. At the input $(1, 2, 4, 8)$
the outputs are $(1.5, 3, 6)$, and applying the composite to the output selection $(1.5, \bot, \bot)$ returns
$(1.5, 3, \bot)$, highlighting that $1.5$ shares an input with $3$; the third output stays $\bot$ because its
window is disjoint. But whereas in these prior works, cognacy queries like these run over an execution record,
here the intermediate structure has been folded away by the chain rule: each derivative is a direct relation
on endpoints.

\section{Correctness of the Higher-Order Interpretation}
\label{sec:definability}

The interpretation of the higher-order language at higher types mixes
the interpretation of the underlying computation with the
interpretation of the semiring dependency matrices
(c.f. \thmref{fam-closed}). It is therefore not immediate that
intepreting a program in $\Fam(\SemiMod(S))$ will give us the expected
interpretation purely on first order data. Note that the fullness of
the functor $H : \Fam(\SDSemiMod(S)) \to \Fam(\SemiMod(S))$ does mean
that we are guaranteed that the tangent maps are correctly computed.
\citet{vakar22} and \citet{nunes2023} construct custom
instances of categorical sconing arguments to prove correctness of
their higher-order interpretation with respect to normal
differentiation. Instead of doing this, we make use of a general
\emph{syntax free} theorem due to \citet{fiore-simpson99}. The proof
of this depends on the construction of a Grothendieck Logical Relation
over the extensive topology on the category $\cat{C}$, but the
statement of the theorem does not rely on this. We have formalised
this proof in Agda (see \texttt{conservativity.agda} in the
supplementary material).

\begin{theorem}[\citet{fiore-simpson99}]
  \label{thm:glr-definability}
  Let $\cat{C}$ be an extensive bicartesian category, $\cat{D}$ be a
  bicartesian closed category, and $F : \cat{C} \to \cat{D}$ a functor
  preserving finite products and coproducts. Then there is a category
  $\GLR(F)$ and functors $p : \GLR(F) \to \cat{D}$ and
  $\hat{F} : \cat{C} \to \GLR(F)$, such that:
  \begin{enumerate}
  \item $\GLR(\cat{D},F)$ is bicartesian closed;
  \item $F = p \circ \hat{F} : \cat{C} \to \cat{D}$;
  \item The functor $p$ strictly preserves the bicartesian closed structure; and
  \item The functor $\hat{F}$ is full and preserves the bicartesian structure.
  \end{enumerate}
\end{theorem}

\begin{remark}
  Compared to the exact result stated at the end of
  \citet{fiore-simpson99}'s paper, we have made two modifications,
  justified by our Agda proof. First, we generalise to the case where
  $\cat{C}$ is not Cartesian closed, and the functor $F$ does not
  preserve exponentials. Examination of the proof reveals that if this
  is the case, then $\hat{F}$ also preserves exponentials, but it is
  not needed for the result stated. Second, Fiore and Simpson restrict
  to the case when $\cat{C}$ is small to be able to construct
  Grothendieck sheaves on this category. We use Agda's universe
  hierarchy to simply construct ``large'' sheaves at the the
  appropriate universe level.
\end{remark}

We use \thmref{glr-definability} to show that the
interpretation of the language in the category $\Set$ agrees with the
higher-order interpretation in $\Fam(\SemiMod(S))$
on the underlying function at first order, as required. This shows that the
higher-order interpretation does what we expect in the underlying
interpretation of terms, and that the approximation information does
not interfere.

\begin{theorem}
  \label{thm:underlying-interp-equal} For all
  $\Gamma \vdash M : \tau$, where $\Gamma$ and $\tau$ are first-order,
  the underlying function in the interpretation
  $\sem{\Gamma \vdash M : \tau}_{\Fam(\SemiMod(S))}$ is equal to the
  interpretation $\sem{\Gamma \vdash M : \tau}_{\Set}$ in $\Set$.
\end{theorem}

\begin{proof}
  Instantiate \thmref{glr-definability} with the functor
  \begin{displaymath}
    \langle \mathrm{Id} , \pi_1 \rangle : \Fam(\SemiMod(S)) \to \Fam(\SemiMod(S)) \times \Set
  \end{displaymath}
  that is the identity in the first component and projects out the
  underlying function in the second. By
  \lemref{pi1-preserve-products-and-coproducts}, this functor
  preserves products and coproducts. For each
  $\Gamma \vdash M : \tau$, we obtain a $g$ such that
  $g = \sem{\Gamma \vdash M : \tau}_{\Fam(\SemiMod(S))}$ and
  $\pi_1(g) = \sem{\Gamma \vdash M : \tau}_\Set$. Substituting $g$
  yields the result.
\end{proof}

\section{Related Work}
\label{sec:related-work}

\paragraph{Automatic Differentiation}

Automatic differentiation (AD), discussed in \secref{approx-as-tangents:autodiff}, is the idea of computing
derivatives of functions expressed as programs by systematically applying the chain rule. The observation that
these derivative computations could be interleaved with the evaluation of the original program is due to
\citet{linnainmaa76}, who showed how the forward derivative $\pushf{f}_x$ of $f$ at a point $x$ could be
computed alongside $f(x)$ in a single pass, dramatically improving the efficiency of derivative evaluation
over symbolic or numerical differentiation. This insight became the foundation of forward-mode AD, which
underpins many optimisation and scientific computing tools, including JAX~\cite{jax2018github}.
\citet{griewank89} showed how the Wengert list, the linear record of assignments used in forward-mode to
compute derivatives efficiently, could be traversed in reverse to compute the pullback map. This two-pass
approach is the foundation of reverse-mode AD, and closely resembles implementations of Galois slicing
(\secref{related-work:galois-slicing} below) that record a trace during forward slicing for use in backward
slicing.

More recent approaches to automatic differentiation have emphasised semantic foundations. \citet{elliott18}
proposed a categorical model of AD that interprets programs as functions enriched with their derivatives,
giving a compositional account of differentiation based on duality and linear maps. Vákár and
collaborators~\cite{vakar22,nunes2023} developed the CHAD framework which inspired this paper, using
Grothendieck constructions over indexed categories to capture both values and their tangents in a
compositional semantic structure. These perspectives shed light on the categorical structure of AD and guide
the design of systems that generalise AD% beyond real analysis
, including the application to data provenance
and slicing explored in this paper.

\paragraph{\GPS}
\label{sec:related-work:galois-slicing}

Galois slicing was introduced by Perera and collaborators~\cite{perera12a,perera13} as an operational approach
to program slicing for pure functional programs, based on Galois connections between lattices of input and
output approximations.  Subsequent work extended
the approach to languages with assignment and exceptions~\cite{ricciotti17} and
concurrency~\cite{perera16d}.
 The present work is instead denotational, recasting dependency analysis as
differentiation over a semiring; the Boolean case recovers Galois slicing, and other semirings give further
analyses. Galois slicing is also more general than our framework in two respects that we plan to consider in
future work (\secref{conclusion}): it computes \emph{program slices} from approximation lattices representing
partially erased programs, and it approximates the \emph{shape} of data rather than only the values at fixed
positions. More recently, \citet{perera22} presented a variant of Galois slicing
over Boolean algebras, which unlike plain lattices are complemented, giving each analysis a
conjugate. Composing an analysis with its conjugate computes
\emph{related outputs}, which they used to support interactive visualisations with linked selections.
\citet{bond25} revisited this using \emph{dynamic dependence graphs}, computing the conjugate from the
opposite graph. The related-outputs analysis is
the cognacy relation shown in \secref{examples:linked}, computed here by composing a Boolean Jacobian with its
transpose.

\paragraph{Semiring Provenance}

Provenance for database queries has its own semiring tradition, originating with the \emph{provenance
semirings} of \citet{green07}: a query result carries an annotation from a commutative semiring, and the free
semiring $\mathbbm{N}[X]$ of provenance polynomials is universal, specialising to other semirings by
evaluation. Our framework draws on the same semirings but computes a \emph{derivative} rather than an
annotation: a Jacobian over the semiring relating input positions to output positions, in place of a value
attached to the output. A derivative is a linear map, so the provenance it records is \emph{linearised}
(\secref{approx-as-tangents:linearised-provenance}), and joint use appears as a coefficient rather than the
degree of a polynomial. Over a non-positive semiring, the Boolean dependency our framework reports is a sound
over-approximation, since the exact analysis can cancel a contribution to zero where the Boolean composite
still records a dependency. Extending provenance to deletion and difference likewise calls for non-positive
annotations~\cite{green09,amsterdamer11,geerts10}.

\paragraph{Discrete and Qualitative Derivatives}

Discrete data has its own differential calculi, developed independently of numerical differentiation.
The \emph{Boolean differential calculus} of \citet{steinbach17} arose in the 1950s from the detection of
faults in digital circuits; its \emph{simple derivative} $\partial f / \partial x_i = f|_{x_i = 0} \veebar
f|_{x_i = 1}$, the exclusive-or of $f$ with the input $x_i$ set to $0$ and to $1$, is $1$ exactly when
flipping that input flips the output. Exclusive-or and conjunction make the Booleans a \emph{ring} rather than
a lattice, so this is a derivative over the two-element field. Because exclusive-or cancels, the derivative is
exact, whereas the join of the Boolean lattice underlying Galois slicing
(\secref{related-work:galois-slicing}) cannot cancel and only reports a possible dependency. The sign domain
gives a second such calculus. The rule of signs of
abstract interpretation~\cite{cousot77} and the \emph{qualitative derivatives} of qualitative
physics~\cite{dekleer84,kuipers86} ask whether increasing an input increases, decreases, or cannot affect an
output; see the sign semiring of \secref{approx-as-tangents:qualitative}. Our account treats the semiring as a
parameter, relating these calculi to automatic differentiation and provenance through the choice of
scalars.

\paragraph{Tangent Categories and Differential Linear Logic}

\emph{Tangent categories}, due originally to \citet{rosický84} and developed by \citet{cockett14,cockett18},
provide an abstract categorical framework for reasoning about differentiation, inspired by the structure of
the tangent bundle in differential geometry. In a tangent category, each object $X$ is equipped with a tangent
bundle $T(X)$, and each morphism $f: X \to Y$ has a corresponding differential map $T(f): T(X) \to T(Y)$
satisfying axioms analogous to the chain rule and linearity of differentiation. Tangent categories generalise
Cartesian differential categories \cite{cdcs}, which model differentiation over Cartesian closed categories
using a syntactic derivative operator. Reverse Tangent categories \cite{reverse-tangents} further axiomatise
the existence of reverse derivatives. Our construction realises differentiation of this kind concretely, over an arbitrary commutative
semiring, and we conjecture that the categories underlying it are examples of tangent, or reverse tangent,
categories. There are likely links to Differential Linear Logic
\cite{ehrhard_differential_2006}. Differential Linear Logic and the Dialectica translation have been used to
model reverse differentiation by \citet{kerjean-pedrot2024}.

\section{Conclusion and Future Work}
\label{sec:conclusion}

We have presented a semantic account of data provenance as a form of automatic differentiation, representing
the dependency of outputs on inputs as a derivative whose Jacobian is a matrix over a commutative semiring.
Traditional automatic differentiation, and various program analysis techniques such as \GPS, arise as
instances. The model interprets an expressive higher-order language for querying and manipulating data, and
reveals new applications such as approximation by intervals~(\secref{approx-as-tangents:perturbation-bounds}).
It also opens up an opportunity to integrate traditional symbolic data provenance techniques with the AD-based
methods used in explainable AI, such as the saliency maps of Grad-CAM~\cite{selvaraju20}. Our categorical
approach admits a modular construction of the model, and the use of general theorems, such as Fiore and
Simpson's \thmref{glr-definability}, to prove properties of the interpretation. We have focused on
constructions that enable an executable implementation in Agda.

\paragraph{Alternative semiring models}

The choice of scalar semiring is the principal degree of freedom in our framework, and choices
beyond those studied here seem worth pursuing. The perturbation-bound semirings of
\secref{examples:perturbation-bounds} point towards metric approximation: a Lawvere metric space is a
category enriched over the min-plus quantale, whose operations are those of our absolute-perturbation
semiring, so metric spaces provide a native notion of approximation already partly within reach of the
framework. Relating this to \citet{edalat-heckmann98}'s computational model of metric spaces is one promising
direction. A second is to replace the setoids carrying values with sets equipped with a semiring-valued
equality relation, over a partially ordered semiring with a residual for its multiplication.

% \paragraph{Categorical Models of Differentation} We have already identified \conref{tangent-stable-fns} and
% \conref{tangent-fam} as potential connections to categorical models of differentation.

\paragraph{Connection to sequentiality and stable functions}

In our model the forward derivative carried by a morphism of $\Fam(\SemiMod(S))$ is supplied as part
of the interpretation, not determined by the underlying function, so it need not reflect how that function
depends on its input. Berry's
\emph{stable functions}~\cite{berry79,berry82}, by contrast, come equipped with an intrinsic forward and
backward derivative. Stability was proposed as a refinement of domain theory aimed at capturing the
intensional behaviour of sequential programs, imposing constraints on how functions preserve bounded meets of
approximants; although stable functions do not fully characterise sequentiality, because they admit
$\mathrm{gustave}$-style counterexamples, they remain an appropriate notion for studying the sensitivity of a
program to partial data at a specific point~\cite{amadio-curien}. Paul Taylor's characterisation of stable
functions via local Galois connections on principal downsets provides a semantic underpinning for the reverse
maps of Galois slicing~\cite{taylor99}. The stable-function story has a tight connection to Galois slicing but a looser one to automatic
differentiation, so we intend to develop it in separate work, relating Galois slicing to stable functions
rather than to differentiation. Such a development would also clarify the links to Differential Linear
Logic~\cite{ehrhard_differential_2006} and to Crubill\'e's identification of stable functions with power
series~\cite{crubille18}, relating the qualitative derivatives studied here to the analytic derivatives of
that setting.

\paragraph{Shape Approximation}

Lists are the only recursive datatype in our source language, so supporting general inductive and
coinductive types is important future work, perhaps adapting \posscite{nunes2023} automatic differentiation
for $\mu\nu$-polynomial functors. Richer datatypes bring the question of \emph{shape approximation}:
approximating not only the values at the positions of a structure but the structure itself, such as the
length of a list or the depth of a tree. Shape approximation features in earlier work on Galois
slicing~\cite{perera12a,perera16d,ricciotti17}, whose approximation lattices represent partially erased data,
whereas our model holds the shape of each input fixed. Our Jacobians are spans, and polynomial functors are
the shape-dependent generalisation of spans~\cite{gambino-kock}; the finite polynomials form the Lawvere
theory for commutative semirings, the algebra of our scalars. The corresponding datatypes are containers,
whose derivative is the type of one-hole contexts, carrying a notion of linear map and a chain rule for both
inductive and coinductive types~\cite{abbott05}. This already gives a set-theoretic account of shape
approximation; weighting positions with our semiring scalars would extend our analyses to varying shapes.

\paragraph{Recursion and Partiality}  This work treats only total programs, and makes no use of directed completeness or similar
properties of domains. Extending to recursion and partiality would introduce a notion of undefinedness, and
hence a definedness order to be reconciled with any approximation order on the same data. Non-termination and
infinite data are both partiality phenomena, approximating unbounded computation by finite observation, so
potentially the finite prefixes of an infinite or non-terminating computation would fall under the same
extension.

\paragraph{Source-To-Source Translation Techniques}

% To formalise this, we'd need to... See also $\Fam(\PSh_{\CMon}(\namedcat{Syn}))$?

An interesting alternative to the denotational approach presented here, and to the trace-based approaches used
in some implementations of provenance, would be to develop a
source-to-source transformation, in direct analogy with
the CHAD approach to automatic differentiation \cite{vakar22,nunes2023}. In their approach, forward and
reverse-mode AD are implemented as compositional transformations on source code, guided by a universal
property: they arise as the unique structure-preserving functors from the source language to a suitably
structured target language formalised as a Grothendieck construction. Adapting this to our setting
would allow the analysis to be ``compiled in'', avoiding the need for
a custom interpreter and potentially exposing opportunities for optimisation.

\bibliographystyle{ACM-Reference-Format}
\bibliography{bib}

% \pagebreak
% \appendix
% \input{appendix/notes}

\end{document}